\documentclass[11pt]{article}
\usepackage{graphics}
\usepackage{authblk}
\usepackage[dvips]{graphicx}
\usepackage[english]{babel}
\usepackage{amsthm}
\newtheorem{theorem}{Theorem}[section]

\usepackage[utf8]{inputenc}
\usepackage{amsmath,amsthm,amssymb}
\usepackage{mathrsfs}
\usepackage[mathscr]{euscript}
\usepackage{graphicx}
 \usepackage{epstopdf}
\usepackage{float}
\usepackage{xcolor}
\usepackage{ulem}
\usepackage{cancel}
\usepackage{wrapfig}
\definecolor{uofsgreen}{rgb}{0.0, 0.8, 0.6}
\definecolor{brickred}{rgb}{0.8, 0.25, 0.33}

\usepackage{subcaption}
\usepackage{graphicx}

\newcommand{\vect}[1]{{\bf {#1} }}

\makeatletter
\numberwithin{equation}{section}
\newcommand{\leastspacing}{\let\CS=\@currsize\renewcommand{\baselinestretch}{.6}\tiny\CS}
\newcommand{\singlespacing}{\let\CS=\@currsize\renewcommand{\baselinestretch}{1}\tiny\CS}
\newcommand{\oneandahalfspacing}{\let\CS=\@currsize\renewcommand{\baselinestretch}{1.4}\tiny\CS}
\newcommand{\onespacing}{\let\CS=\@currsize\renewcommand{\baselinestretch}{1.2}\tiny\CS}
\newcommand{\oneh}{\let\CS=\@currsize\renewcommand{\baselinestretch}{1.35}\tiny\CS}
\newcommand{\doublespacing}{\let\CS=\@currsize\renewcommand{\baselinestretch}{1.7}\tiny\CS}
\oneh
\newcommand{\be}{\begin{equation}}
\newcommand{\ee}{\end{equation}}
\newcommand{\beano}{\begin{eqnarray*}}
\newcommand{\eeano}{\end{eqnarray*}}
\newcommand{\ba}{\begin{eqnarray}}
\newcommand{\ea}{\end{eqnarray}}
\newcommand{\no}{\nonumber}

\newcommand{\hone}{\mbox{\hspace{1em}}}

\def\pr{\partial}

\begin{document}
\bibliographystyle{unsrt}
%\singlespacing
%\doublespacing
\title{A new model for two-layer liquid-gas stratified flows in pipes with general cross sections}

\author[1]{Sarswati Shah%
	\thanks{Email: \texttt{sshah66@gmu.edu}}}
\author[2]{Gerardo Hern\'andez-Due\~nas\thanks{Email: \texttt{hernandez@im.unam.mx}}}
\affil[1]{{Center for Mathematics and Artificial Intelligence
and Department of Mathematical Sciences, George Mason University, Fairfax, VA 22030, USA}}
\affil[2]{{Institute of Mathematics, National University of Mexico, Blvd. Juriquilla 3001, Queretaro, Mexico}}
\date{}

%Two-layer stratified fluid flows
%\author[1]{Sarswati Shah%
	%\thanks{Email: \texttt{sarwatishah@im.unam.mx}}}
%\author[1]{Gerardo Hern\'andez-Due\~nas\thanks{Email: \texttt{hernandez@im.unam.mx}}}
%\affil[1]{\textit{Institute of Mathematics, National University of Mexico, Blvd. Juriquilla 3001, Queretaro, Mexico}}
%\date{}
\maketitle

\begin{abstract}
In this work, we derive a new model for immiscible two-layer gas-liquid stratified flows in pipes with general cross sections. The bottom layer is occupied by an incompressible fluid in liquid phase with hydrodynamics based on a hydrostatic pressure, following a shallow water approximation. The top layer is occupied by a compressible gas, following an ideal gas law leading to conservation of mass, momentum and energy. The two subsystems are linked through non-conservative products, representing momentum and energy exchanges between layers. The hyperbolic properties of the resulting model are analyzed, including the derivation of entropy inequalities, and the approximations of eigenvalues of the corresponding coefficient matrix. Numerical tests are included to demonstrate the merits of the model and the numerical approximations, including well-balancedness, Riemann problems, and perturbations and convergence toward steady states at rest. Besides simulations of water and air where the density difference between layers is significant, a case where such difference is not so pronounced (like gas and liquid hydrogen) is also shown. 
\end{abstract}


\section{Introduction}

The handling of gas–liquid mixtures using centrifugal pumps has grown significantly, particularly in hydraulic, nuclear, and chemical engineering applications \cite{chaudhry1990analysis}. However, the complexities of gas–liquid flows \cite{poullikkas2003effects} within pump impellers make it difficult to estimate fluctuations in pump performance. The presence of both gas and liquid phases in these pipelines can lead to pressure surges and the formation of moving air pockets during the flow evolution. Formulating an accurate mathematical model to discuss these phenomena is a complex task, as they may exhibit different regimes in the case of multi-phase flows (see \cite{pai2013two}). This may happen, for instance, in a pressurized flow where the pipe is full of liquid (slug flow). Additionally, stratified flow may occur, where the interface between the fluids can be either wavy or flat, depending on factors like the gas flow rate, among others. Variations in cross-sectional area induce strong accelerations and decelerations of the mixture and modify the local pressure distribution, which can in turn promote transitions between different flow regimes and significantly influence the onset of pressure surges and air-pocket dynamics. In this context, numerical and experimental analysis are presented in \cite{besharat2016study} to control such pressure surges in water networks. See also \cite{banda2006coupling}, where gas flow is modeled in pipeline networks using the isothermal Euler equations, and \cite{brouwer2011gas} where a hierarchy of models is presented, considering different parameter regimes.  

These physical and geometric complexities described above motivate us to develop a mathematical model that can account for variations in cross-sectional area and compressibility effects. More specifically, this manuscript focuses on modeling two-phase flow in pipes with varying cross-sections. We adopt a combined methodology to analyze the dynamics of two-layer stratified flows in pipes. The bottom layer is assumed to be occupied by an incompressible fluid in liquid phase, driven by a hydrostatic pressure, with a time evolution that follows a shallow-water approximation. The top layer is assumed to be occupied by a compressible gas, evolving according to an ideal gas law and an equation of state used for the Euler equations. The two sub-systems are connected through non-conservative products that result from momentum and energy exchanges between the two layers. 

Two-layer liquid-gas stratified flows in ducts have been studied in previous works. For instance, a model for such gas-liquid interactions was derived in \cite{demay2017compressible} for channels with rectangular or circular cross sections, with compressible fluids in both layers. A key distinction of our work here from the above reference lies in our assumption that the bottom layer is completely incompressible while the top one is fully compressible. In \cite{bourdarias2013air}, a similar approach to ours was followed, computing averages and considering interactions between incompressible and compressible fluids. However, they consider an isentropic flow under an isothermal process while we use a different equation of state for ideal fluids. On the other hand, the internal geometry of the conduit is rarely uniform in practical pipeline engineering applications. Pipes may gradually contract or expand, or may include localized features such as obstacles and valves \cite{capart1997numerical}.  This work is novel in the current literature and can be applied to pipes of any shape, whether shrinking or expanding, and regardless of the cross sectional geometry.

Addressing two-phase flow in this setting involves several theoretical and numerical challenges. The interaction between the liquid and gas layers generates additional terms that cannot be written purely in conservation form, but instead appear as non-conservative products in the governing equations. Consequently, the system falls outside the classical framework of conservation laws, and the notion of weak solution must be treated with care. Typical attempts include prescribing suitable paths \cite{dal1995definition} in state space across discontinuities. The choice of path is pivotal and, ideally, should be informed by the underlying physics. For layer-averaged shallow-water models, no canonical prescription exists; distinct admissible paths may yield different weak solutions, introducing a degree of arbitrariness. In multiphase contexts, for example, \cite{raviart1995nonconservative} defines the path via the traveling-wave profile of a viscously regularized system, whereas \cite{toumi1992weak} adopts the Volpert (straight-line) path. Entropy inequalities may be helpful in selecting physically meaningful solutions. In this spirit, we derive an entropy pair for the proposed model and prove the associated entropy inequality, thereby providing a robustness property for the formulation.

From a numerical perspective, the presence of non-conservative terms complicates the construction of robust finite-difference and finite-volume schemes. In this setting, path-conservative methods were introduced in \cite{pares2006numerical}, building on the theoretical framework of \cite{dal1995definition}. However, it was shown in \cite{abgrall2010comment} that such schemes may fail, in general, to compute the correct solutions, even when the appropriate path is assumed to be known. Further developments include entropy-stable, path-consistent schemes \cite{castro2013entropy} or relaxation approaches for coupled systems \cite{kolbe2024numerical,abgrall2009two}. Nevertheless, a general uniqueness theory remains an open problem \cite{bianchini2005vanishing}. Instead of adopting a path-conservative discretization, we focus on the implementation of a central-upwind scheme designed to preserve key structural properties, namely hyperbolicity, well-balancing, and positivity of the liquid-layer depth \cite{kurganov2009central, hernandez2021central}.

The rest of the manuscript is organized as follows. Section \ref{sec:TheModel} describes the model formulation and its derivation through cross-sectional averaging together with the corresponding assumptions and boundary conditions. It explains the model hyperbolic properties, including the existence of an entropy pair.  Section \ref{sec:NumResults} includes a variety of numerical tests that exhibit the robustness of the model and the numerical algorithm.

\section{Model formulation and its properties}

\label{sec:TheModel}

\subsection{Basic equations}
\label{sec:BasicEqns}

In order to derive our model, we start with the following three-dimensional set of equations for conservation of mass and balance of momentum. That is taken into consideration for both layers \cite{saurel2018diffuse, lemartelot2013liquid, demay2017compressible, demay2019splitting}. In addition, the top layer will follow an equation of state for ideal gases. As a result, we require an equation for conservation of energy to govern the upper layer. We then have
\begin{subequations} \label{Euler}
\begin{align}
\hspace{-0.35cm}\frac{\pr \rho_k}{\pr t} + \nabla \cdot (\rho_k \textbf{u}_k) = 0 \hone \mbox{(Mass conservation for $k=1,2$)}, \label{1a}\\
\hspace{-0.35cm}\frac{\pr \rho_k \textbf{u}_k}{\pr t} + \nabla \cdot (\rho_k \textbf{u}_k \otimes \textbf{u}_k + \mathcal{P}_k I) =  \rho_k \mbox{\textbf{g}} \hone \mbox{(Momentum cons., $k=1,2$)}, \label{1b}\\
\hspace{-0.35cm}\frac{\pr \rho_2 \mathcal{E}_2}{\pr t} + \nabla \cdot \left((\rho_2 \mathcal{E}_2 + \mathcal{P}_2)\textbf{u}_2 \right) =0 \hone \mbox{(Energy conservation for $k=2$)}, \label{1c}
\end{align}
\end{subequations}
where the indices $k=1,2$ denote the lower layer and upper layer, respectively. The velocity vector of each phase is given by the variable $\textbf{u}_k = (u_k, v_k, w_k)$, the density by $\rho_k$, $\textbf{g}$ is the acceleration due to gravity, $\mathcal P_k$ the pressure, $ \mathcal{E}_2 = e_2  + \frac{1}{2} {u_2}^2$ is the gas energy density, with $e_2$ being the internal gas energy density.

\begin{figure}[htp]
\centering
{\includegraphics[width=.55\textwidth]{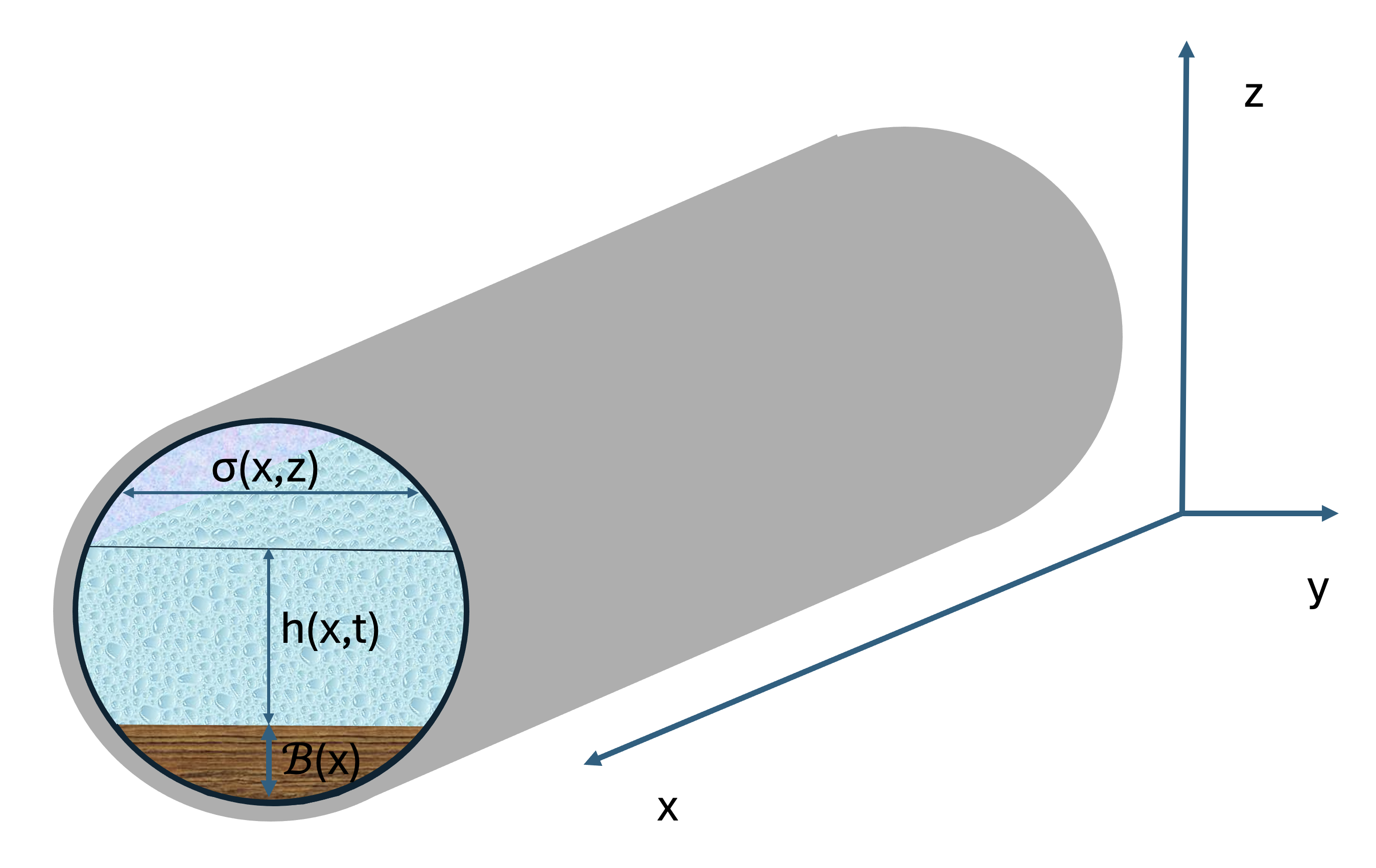}}
{\includegraphics[width=.44\textwidth]{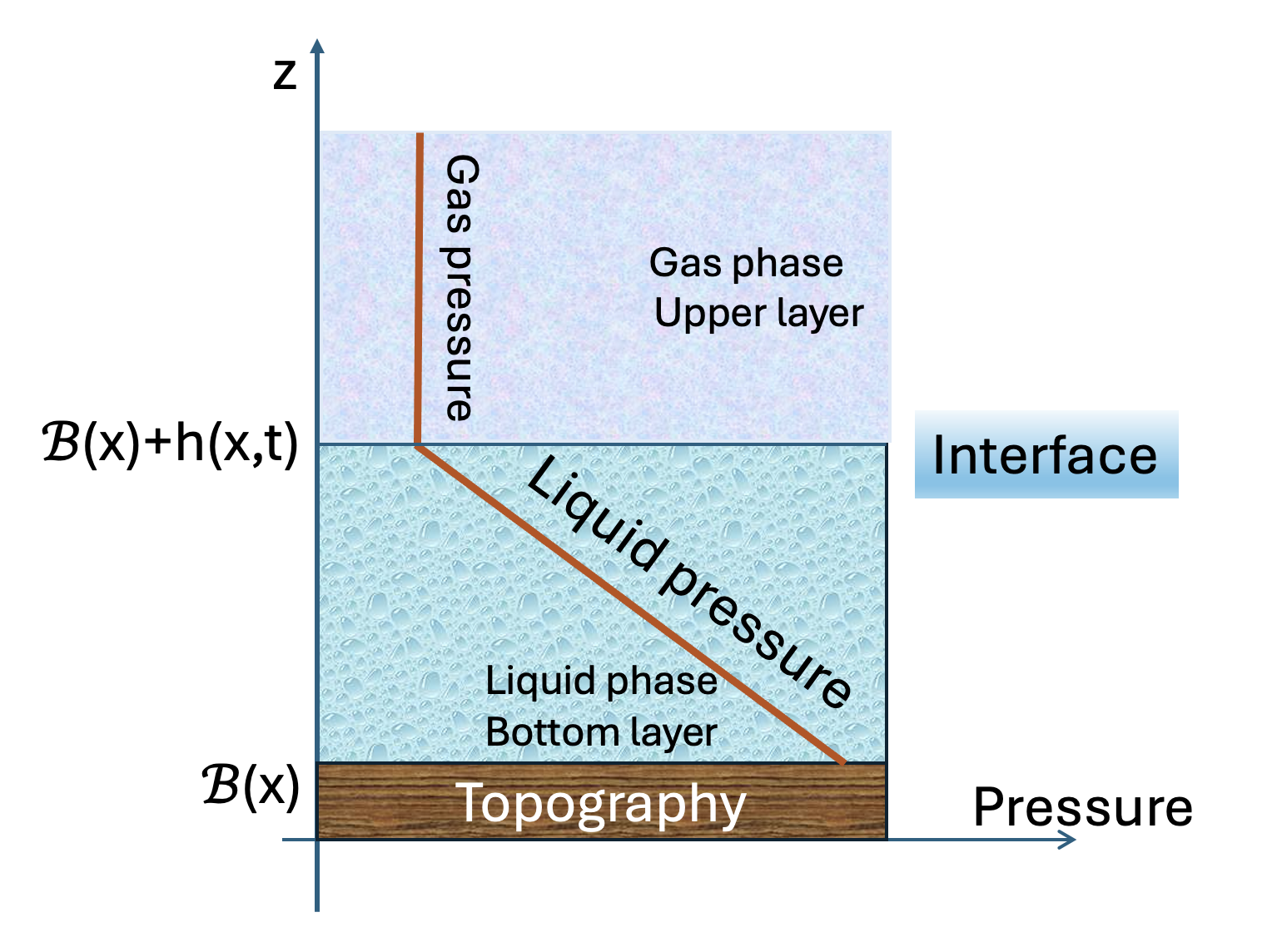}}
\caption{Schematic. Left panel: 3D schematic of the duct's geometry. Right panel: pressure as a function of height. }
\label{fig:schematic}
\end{figure}

\paragraph{Duct's geometry}
In our model (to be derived in the next section), we will assume that the bottom and upper layers are separated by an interface. The model computes the time evolution of the cross-sectional averages of the corresponding conserved variables (such as momentum and energy). Such averaging process involves integrating the equations of motion over the cross sectional region that is occupied by each layer. The averaging process eliminates information about variations in the cross-sectional directions, but it can still provide accurate predictions of the fluid’s propagation along the duct’s axial direction. We thus assume that the interface that separates both layers is given at height 
\[
z = \mathcal{B}(x)+h(x,t),
\]
where $x$ denotes the duct's axial position, $\mathcal{B}(x)$ the bottom topography and $h(x,t)$ the depth of the bottom layer at time $t$ associated with the cross section at axial position $x$. See the left panel of Figure \ref{fig:schematic} for a schematic of the duct's geometry. The cross section at axial position $x$ can have an arbitrary geometry, given by the region
\[
0 \le z \le H(x), \; -\sigma(x,z) \le y \le \sigma(x,z)/2,
\]
where $\sigma(x,z)$ is the width of the channel at axial position $x$ and height $z$; and $H(x)$ is the maximum channel's height at axial position $x$. 

\paragraph{Pressure}
The fluid in the upper layer is assumed to behave as an ideal compressible gas. The pressure $\mathcal{P}_2$ is assumed to stay pretty uniform in that layer, given ultimately by a function of density $\rho_2$ and internal energy density $e_2$ through an equation of state.  That is, in the limit $\mathcal{P}_2$ varies only in $x$ and $t$ (independent of $y$).

On the contrary, the liquid in the bottom layer is assumed to be governed as an incompressible stratified flow. We adopt the classical hydrostatic assumption for pressure, giving
\ba
\frac{d \mathcal{P}_1}{dz} = -g \rho_1,
\label{hydros}
\ea
where the density $\rho_1$ is assumed to be constant both in space and time. We also impose the following interfacial condition for pressure:
\ba
\mathcal{P}_1|_{z = \mathcal{B}+h} = \mathcal{P}_2. \no
\ea
That is, the pressure is continuous accross the interface. Integrating  \eqref{hydros}, we get
\ba
\label{eq:HydrostaticPressure}
\mathcal{P}_1 = {\mathcal{P}_2} + g \rho_1 ( \mathcal{B}+h - z). \label{HP}
\ea
This pressure as a function of height is depicted in the right panel of Figure \ref{fig:schematic}.

\subsection{Saint-Venant averaging}
Pursuing the Saint-Venant modelling approach \cite{audusse2005multilayer}, we assume that the variations in each cross section is weak, dominated by the motion in the duct's axial direction. Any quantity $\mathcal{F}$ in the model is then averaged in each cross section as
\begin{equation}
\label{depth}
\begin{array}{lcl}
\overline{\mathcal{F}_1(x,t)} & = & \frac{1}{A_1} \int_{\mathcal{B}(x)}^{\mathcal{B}(x)+h(x,t)}\int_{-\sigma/2}^{\sigma/2} \mathcal{F}(x,y,z,t) dydz, \\  \\\overline{\mathcal{F}_2(x,t)} & = & \frac{1}{A_2} \int_{\mathcal{B}(x)+h(x,t)}^{H(x)}\int_{-\sigma/2}^{\sigma/2} \mathcal{F}(x,y,z,t) dydz.
\end{array}
\end{equation}
The term $A_k(x,t)$ denotes the cross-sectional area at axial position $x$ and time $t$ at layer $k$, and it is given by 
\ba
A_1(x,t) = \int_{\mathcal{B}(x)}^{\mathcal{B}(x)+h(x,t)} \sigma(x,z)dz, A_2(x,t) = \int_{\mathcal{B}(x)+h(x,t)}^{H(x)} \sigma(x,z)dz.  \label{VCS}
\ea
We also introduce the density weighted averaging operator as in \cite{demay2017compressible}, giving
\ba
\widehat{\mathcal{F}_k(x,t)} = \frac{\overline{\rho_k \mathcal{F}}}{\overline{\rho_k }}.
\ea
Integrating the continuity equation, we get (for any $z_1$ and $z_2$)
\ba
\int_{z_1}^{z_2}\int_{-\sigma/2}^{\sigma/2} \frac{\pr \rho_k}{\pr t} dy dz + \int_{z_1}^{z_2} \int_{-\sigma/2}^{\sigma/2} \frac{\pr \rho_k u_k}{\pr x} dy dz + \int_{z_1}^{z_2} \int_{-\sigma/2}^{\sigma/2} \frac{\pr \rho_k  v_k}{\pr y} dy dz \no\\
 + \int_{z_1}^{z_2} \int_{-\sigma/2}^{\sigma/2} \frac{\pr \rho_k  w_k}{\pr z} dy dz = 0. \label{MA} 
\ea
For the lower layer, $z_1=\mathcal{B}(x)$ and $z_2=\mathcal{B}(x)+h(x,t)$. Using the fundamental theorem of calculus and the Reynolds transport theorem, the  averaged continuity equation for the lower layer (i.e. $k=1$) becomes
\ba
\label{eq:AvgMass}
\frac{\pr}{\pr t}(A_1 \overline{\rho_1}) + \frac{\pr}{\pr x}(A_1 \overline{\rho_1} \; \widehat{u_1}) - \int_{-\sigma/2}^{\sigma/2} \rho_1 u_1 dy\bigg|_{z=\mathcal{B}+h} \frac{\pr}{\pr x}(\mathcal{B}+h) + \int_{-\sigma/2}^{\sigma/2} \rho_1 u_1 dy\bigg|_{z=\mathcal{B}} \frac{\pr \mathcal{B}}{\pr x}\no\\
-\int_{\mathcal{B}}^{\mathcal{B} + h}\rho_1 u_1\bigg|_{y=\sigma/2}\frac{\pr}{\pr x}\left(\frac{\sigma}{2} \right)dz + \int_{\mathcal{B}}^{\mathcal{B} + h}\rho_1 u_1\bigg|_{y=-\sigma/2}\frac{\pr}{\pr x}\left(\frac{-\sigma}{2} \right)dz + \int_{\mathcal{B}}^{\mathcal{B} + h}\rho_1 v_1\bigg|_{y=\sigma/2}dz \no\\ 
- \int_{\mathcal{B}}^{\mathcal{B} + h}\rho_1 v_1\bigg|_{y=-\sigma/2} dz + \int_{-\sigma/2}^{\sigma/2} \rho_1 w_1 dy\bigg|_{\mathcal{B}}^{\mathcal{B}+h} - \int_{\mathcal{B}}^{\mathcal{B} + h}\rho_1 w_1\bigg|_{y=\sigma/2} \frac{\pr}{\pr z}\left(\frac{\sigma}{2} \right)dz \no\\ 
+ \int_{\mathcal{B}}^{\mathcal{B} + h}\rho_1 w_1\bigg|_{y=-\sigma/2} \frac{\pr}{\pr z}\left(\frac{-\sigma}{2} \right)dz = 0.\hspace{0.31cm}
\ea

\paragraph{Streamline\;conditions}

In equation \eqref{eq:AvgMass}, the conserved variable is now $A_1 \overline{\rho_1}$, with flux $A_1 \overline{\rho_1} \; \widehat{u_1}$. There are additional boundary terms that arise during the averaging process. However, those terms disappear if the interface and duct's walls are streamlines. These boundary conditions imply, for instance, that fluid particles located at the interface remain there. As a result, the interface offers no resistance to the flow. Specifically, it can be stated as follows. The interface is parametrized by $\vec{r}(x,y,t) = (x,y, \mathcal{B}(x)+h(x,t) )$. A vector normal to the surface satisfies $ \vec{n} \sim \vec{r}_x \times \vec{r}_y =  \left(-\frac{\pr (h +\mathcal{B})}{\pr x},0,1 \right) $. The relative vertical velocity of the fluid is $w-\partial_t h$, where $\partial_t h$ is the surface time variation. Its normal component to the surface must vanish. So,
\ba
0 = (u,v,w-\partial_t h)\cdot \left(-\frac{\pr (\mathcal{B}+h)}{\pr x},0,1 \right) \big |_{z = \mathcal{B}+h}, 
\ea
which gives
\ba
w \big |_{z=\mathcal{B}+h} = \partial_t h + u \big |_{z = \mathcal{B}+h} \frac{\pr (h +\mathcal{B})}{\pr x}. \label{SC1}
\ea

By symmetry, we assume that the walls are parametrized by $(x,z) \longrightarrow (x,y=\pm \sigma/2,z)$, with normal vectors $\vec{n} \sim  \left(\pm \frac{\sigma_x}{2}, -1, \pm \frac{\sigma_z}{2} \right)$. By following the same streamline conditions over the walls, we get to
\ba
\label{SC2}
\begin{array}{lcl}
v_1|_{y=\frac{\sigma}{2}} & = & u_1 \frac{\pr}{\pr x} \left(\frac{\sigma}{2} \right) + w_1 \frac{\pr}{\pr z} \left(\frac{\sigma}{2} \right)\bigg|_{y=\frac{\sigma}{2}}, \\
v_1|_{y=\frac{-\sigma}{2}} & = & u_1 \frac{\pr}{\pr x} \left(-\frac{\sigma}{2} \right) + w_1 \frac{\pr}{\pr z} \left(-\frac{\sigma}{2} \right)\bigg|_{y=\frac{-\sigma}{2}}.
\end{array}
\ea

\paragraph{Averaging process}
Using the streamline conditions (\ref{SC1}, \ref{SC2}), we get the final averaged equation for conservation of mass at the bottom layer:
\ba
\boxed{\frac{\pr}{\pr t}(A_1 \overline{\rho_1}) + \frac{\pr}{\pr x}(A_1 \overline{\rho_1}\; \widehat{u_1})=0}. \label{MAFL}
\ea
As it was mentioned above, the fluid of the bottom layer is assumed to be incompressible with $\overline{\rho_1}$ constant both in space and time.

The equation for conservation of mass is obtained analogously, giving 
\ba
\boxed{\frac{\pr}{\pr t}(A_2 \overline{\rho_2}) +  \frac{\pr}{\pr x}(A_2 \overline{\rho_2}\; \widehat{u_2})= 0}. \label{MAF}
\ea

\vspace{0.2cm}
\noindent
The integration of equation \eqref{1b} over each cross section and layer, together with the use of the Reynolds theorem, leads to
\ba
\frac{\pr}{\pr t} (A_1 \overline{\rho_1}\widehat{u_1}) + \frac{\pr}{\pr x}\left(A_1 (\overline{\rho_1}\widehat{u_1^2} + \overline{\mathcal{P}_1})\right) - \int_{-\sigma/2}^{\sigma/2}\mathcal{P}_1 dy \bigg|_{z=\mathcal{B}+h}  \frac{\pr}{\pr x}(\mathcal{B}+h) + \int_{-\sigma/2}^{\sigma/2}\mathcal{P}_1 dy \bigg|_{z=\mathcal{B}}  \frac{\pr \mathcal{B}}{\pr x}\no\\
- \int_{\mathcal{B}}^{\mathcal{B}+h} \mathcal{P}_1|_{y=\sigma/2}\frac{\pr}{\pr x}\left(\frac{\sigma}{2} \right)dz + \int_{\mathcal{B}}^{\mathcal{B}+h} \mathcal{P}_1|_{y=-\sigma/2}\frac{\pr}{\pr x}\left(\frac{-\sigma}{2} \right)dz =  0. \hspace{0.8cm}\label{momentum} 
\ea
\ba
\frac{\pr}{\pr t}(A_2 \overline{\rho_2}\widehat{u_2}) + \frac{\pr}{\pr x}\left(A_2 (\overline{\rho_2} \widehat{u_2^2} + \overline{\mathcal{P}_2}) \right) -  \mathcal{P}_2|_{z=H(x)} \sigma_2 \frac{\pr}{\pr x}(H(x))+ \mathcal{P}_2|_{z=\mathcal{B}+h} \sigma_1 \frac{\pr}{\pr x}(\mathcal{B}+h) \no\\
=   \int_{\mathcal{B}+h}^{H(x)} \mathcal{P}_{2_W}\frac{\pr \sigma}{\pr x}dz. \no
\ea

Integrating equation \eqref{1c} over the bottom layer of each cross section, and using the streamline conditions, we arrive at
\ba
0 = \frac{\pr}{\pr t} (A_2 \overline{\rho_2}\widehat{\mathcal{E}_2}) + \frac{\pr}{\pr x} \left(A_2 \overline{\rho_2}\widehat{\mathcal{E}_2} + A_2 \overline{\mathcal{P}_2 u_2} \right) - \mathcal{P}_2|_{z = \mathcal{B}+h} \sigma \pr_t h.
\ea
Using equation \eqref{MAFL}, and the fact that $\overline{\rho_1}$ is constant both in space and time, the above equation can be written as
\ba
\boxed{\frac{\pr}{\pr t} (A_2 \overline{\rho_2}\widehat{\mathcal{E}_2}) + \frac{\pr}{\pr x} \left(A_2 \overline{\rho_2} \widehat{\mathcal{E}_2} + A_2 \overline{\mathcal{P}_2 u_2}) \right) + \mathcal{P}_2|_{z = \mathcal{B}+h} \pr_x ( A_1 \hat u_1) = 0. }  \label{energy}
\ea
Here, the last term makes this equation non conservative. We recall that the total energy density $\widehat{\mathcal{E}_2}$ of upper layer is the sum of internal energy density and kinetic energy density, which writes:
\ba
\widehat{\mathcal{E}_2} = \widehat{e_2}  + \frac{1}{2}\widehat{u_2}^2. \label{energy2}
\ea
The pressure in the upper layer is given by an equation of state for an ideal gas:
\ba\label{eq:AvgEoS}
\widehat{e_2} = \frac{\overline{\mathcal{P}_2}}{\overline{\rho_2}(\gamma_2 -1)}, \label{EoS}
\ea
where $\gamma_2$ is a ratio of gas constants.

%---------subsection------------
\subsection{The model}

In Section \ref{sec:BasicEqns}, we have discussed the basic equations of motion before the averaging process. We assumed an upper layer occupied by an ideal compressible gas. The averaged gas pressure is given in terms of the density and internal energy by equation \eqref{eq:AvgEoS}. We also assumed that the pressure varies in the axial direction but it only weakly varies in the cross-sectional direction, so that $\mathcal{P}_2 \approx \overline{\mathcal{P}_2}$. In contrast, the liquid in the bottom layer is assumed to be incompressible with a hydrostatic pressure given by equation \eqref{eq:HydrostaticPressure}. The pressure $\mathcal{P}_1$ increases with depth, and coincides with the upper pressure at the interface. This is depicted in the right hand side of Figure \ref{fig:schematic}.

With the above considerations, we close the reduced system with the following boundary evaluations of pressure:
\[
\begin{array}{llllll}
\mathcal{P}_2 \big |_{z = H(x)} =  \overline{\mathcal{P}}_2, & \mathcal{P}_2 \big |_{z= B+h} = \overline{\mathcal{P}}_2,\\
\mathcal{P}_1 \big |_{z = B+h} = \overline{\mathcal{P}}_2, & \mathcal{P}_1 \big |_{y = \pm \sigma/2} = \overline{\mathcal{P}}_2+g \rho_1 (h+B-z), & \mathcal{P}_1 \big |_{z = B} = \mathcal{P}_2 + g \rho_1 h. \; 
\end{array}
\] 
Furthermore, we assume $\widehat{u_k^2} \approx (\widehat{u}_k)^2$ and $\overline{\mathcal P_2 u_2 } \approx \overline{\mathcal P_2} \;\overline{u_2}$, which is valid for flows that move mainly in the axial direction. The average pressure for the liquid layer writes
\ba
\int_{\mathcal{B}}^{\mathcal{B}+h} \int_{-\sigma/2}^{\sigma/2} \mathcal{P}_1 dy dz = \int_{\mathcal{B}}^{\mathcal{B}+h} \int_{-\sigma/2}^{\sigma/2}\overline{\mathcal{P}_2} dy dz + \int_{\mathcal{B}}^{\mathcal{B}+h} \int_{-\sigma/2}^{\sigma/2}\rho_1 (h + \mathcal{B} - z) dydz. \no
\ea
Since $\overline{\mathcal{P}_2}$ is independent of $y$ and $z$, dividing by $A_1$ we obtain
\ba 
\label{HP1}
\overline{\mathcal{P}_1} = \overline{\mathcal{P}_2} + \frac{g}{A_1}\int_{\mathcal{B}}^{\mathcal{B}+h} \rho_1 (h + \mathcal{B} - z)\sigma(x,z) \; dz. \label{HP1}
\ea

We note that the total area occupied by both layers is known, giving
\ba
\label{eq:A1A2}
A_1 + A_2 = A_{\mathcal{T}}(x) = \int_{\mathcal{B}(x)}^{H(x)} \sigma(x,z) dz \hspace{0.21cm} \mbox{(fixed in time)}.
\ea

Omitting the bar and hat signs, and taking the limit $\widehat{u_k^2} = (\widehat{u}_k)^2$ and $\overline{\mathcal P_2 u_2} = \overline{\mathcal P_2} \overline{u_2}$, our model becomes
\begin{subequations}\label{FEF1}
\begin{align}
\frac{\pr}{\pr t}(A_1 \rho_1) + \frac{\pr}{\pr x}(A_1 \rho_1 u_1)= 0, \label{FEF1a}\\
\hspace{-0.6cm}  \frac{\pr}{\pr t} (A_1 \rho_1 u_1) + \frac{\pr}{\pr x}\left(A_1 ( \rho_1 u_1^2 + \mathcal{P}_1)\right) - \mathcal{P}_{2} \; \sigma_1(x,t)  \; \frac{\pr}{\pr x}(\mathcal{B}+h) + \sigma_\mathcal{B}(x) \; ( \mathcal{P}_2+ g \rho_1 h ) \;  \frac{\pr \mathcal{B}}{\pr x} \no\\
- \int_{\mathcal{B}}^{\mathcal{B}+h} \left( \mathcal{P}_2 + g \rho_1 (\mathcal{B}+h-z) \right) \frac{\pr \sigma(x,z)}{\pr x} dz =  0, \label{FEF1b}\\
\frac{\pr}{\pr t}(A_2 \rho_2) +  \frac{\pr}{\pr x}(A_2 \rho_2 u_2)= 0,\label{FEF1c}\\
\hspace{-0.6cm} \frac{\pr}{\pr t}(A_2 \rho_2 u_2) + \frac{\pr}{\pr x}\left(A_2 ( \rho_2 u_2^2 + \mathcal{P}_2) \right) - \mathcal{P}_{2} \; \sigma_2 \frac{\pr}{\pr x}(H(x)) + \mathcal{P}_{2} \; \sigma_1(x,t) \frac{\pr}{\pr x}(\mathcal{B}+h) \no\\
=  \mathcal{P}_{2} \; \int_{\mathcal{B}+h}^{H(x)} \frac{\pr \sigma(x,z)}{\pr x}dz, \label{FEF1d}\\
\frac{\pr}{\pr t}(A_2 \rho_2 \mathcal{E}_2) + \frac{\pr}{\pr x} \left(A_2 ( \rho_2 \mathcal{E}_2+ \mathcal{P}_2) u_2\right) + \mathcal{P}_2 \; \pr_x ( A_1  u_1) =0, \label{FEF1e}
\end{align}
\end{subequations}
where
\begin{subequations}
\begin{align}
\mathcal{E}_2 & = e_2  + \frac{1}{2}u_2^2, e_2 = \frac{\mathcal{P}_2}{\rho_2 (\gamma_2 -1)}, \\
\rho_1 & =  \text{constant both in space and time}, \\
\mathcal{P}_1 & =  \mathcal{P}_2 + \frac{g \rho_1}{A_1}\int_{\mathcal{B}}^{\mathcal{B}+h}  (h + \mathcal{B} - z)\sigma(x,z) \; dz, \\
\sigma_{\mathcal{B}} & = \sigma(x,z = \mathcal{B}(x)),  \; \sigma_1 = \sigma(x, z = \mathcal{B}+h), \; \sigma_2 = \sigma(x,z = H(x)), \\
A_1 + A_2 & = A_{\mathcal{T}}(x) = \int_{\mathcal{B}(x)}^{H(x)} \sigma(x,z) dz  \hspace{0.21cm} \mbox{(fixed in time)}.
\end{align}
\end{subequations}

Due to equation \eqref{eq:A1A2}, the upper area $A_2$ is then a function of $A_1$. Therefore, our system \eqref{FEF1} has 5 equations and 5 degrees of freedom, closing it. It is worth noting that there is no explicit equation of motion for the interface. Instead, it evolves according to conservation of mass. Equation  \eqref{FEF1a} determines the evolution of $A_1$, which in turn gives $h$ through equation \eqref{VCS}. Equation \eqref{FEF1c} determines the evolution of $\rho_2$ in time. Finally, equations \eqref{FEF1b}, \eqref{FEF1d} and \eqref{FEF1e} determines the time evolution of gas energy, liquid and gas momenta. 

An immediate consequence of the system's structure is that it admits steady states of rest, satisfying
\begin{equation}
\label{eq:SteadyStatesRest}
\mathcal B + h = \text{constant}, \; u_1 = u_2 = 0, \; \rho_2 \mathcal E_2 = \text{constant}. 
\end{equation}
This can be shown via simple manipulations of the momentum equations.

%-------subsection----
\subsection{Energy conservation and entropy inequality} 

System \eqref{FEF1} consists of a two-layer flow with conservation of mass and balance of momentum for both layers plus an energy conservation equation of the upper layer. The two sub-systems for each layer communicate with each other via non-conservative products that allow for momentum and energy exchanges between layers. The bottom layer assumes a shallow-water approximation with no energy equation, which can still be computed as a consequence of the first two. By manipulating the system, one can obtain an entropy pair and entropy inequality, as stated in the following theorem. 

\begin{theorem}
Let us consider system \eqref{FEF1}. Then the entropy $\eta$ and entropy flux $q$ given by
\ba \label{entropy}
\eta = \frac{1}{2}A_1 \rho _1 u_1^2 + g\rho_1 A_1(h + \mathcal{B}) - \rho_1 I_1 + A_2 {\rho_2}{\mathcal{E}_2}, \no\\
q = \frac{1}{2}A_1 \rho _1 u_1^3  + g\rho_1 A_1 u_1(h + \mathcal{B}) + A_2 ({\rho_2}{\mathcal{E}_2}+{\mathcal{P}_2}){u_2}  + \mathcal{P}_{2} A_1 {u_1},
\ea
where $I_1 =  g \int_{ \mathcal{B}}^{ \mathcal{B}+h} (h +  \mathcal{B} -z)\sigma dz$, define an entropy pair. That is, $\eta$ is a convex function of the conserved variables, and satisfy the following inequality
\ba \label{Energy}
\pr_t \eta + \pr_x q, \; \leq  \; 0 \text{ in the weak sense.}
\ea
\end{theorem}

\begin{proof}
A straightforward calculation shows that the Hessian matrix $H_\eta$ of the entropy $\eta$ is given by:
\[
H_\eta = 
\begin{pmatrix}
\frac{u_1^2}{\rho_1 A_1} + \frac{g}{\rho_1 \sigma_1} & \frac{-u_1}{\rho_1 A_1} & 0_{1\times 3}\\
\frac{-u_1}{\rho_1 A_1} & \frac{1}{\rho_1 A_1} & 0_{1\times 3}\\
0_{3\times 1} & 0_{3\times 1} & 0_{3\times 3}\\ 
\end{pmatrix}.
\]
This is a symmetric matrix with associated quadratic form $\frac{g}{\rho_1 \sigma_1} V_1^2+\frac{1}{\rho_1 A_1} (V_2-u_1 V_1)^2$. The entropy is then a convex function of the conserved variables.

As usual, we consider the viscous version of system \eqref{FEF1}, as it is done in \cite{castro2013entropy}. That is, we add terms $\nu_1 \partial_x ( A_1 \rho_1 \partial_x  u_1)$ and $\nu_2 \partial_x (  A_2 \rho_2 \partial_x u_2)$ to the momentum equations, with $\nu_1, \nu_2 > 0$. Straightforward calculations show that
\[
\begin{array}{lcl}
\partial_t \left( A_1 \rho_1 \frac{1}{2} u_1^2 \right) & = & u_1 \partial_t (A_1 \rho_1 u_1 )- \frac{1}{2}u_1^2 \partial_1 (A_1 \rho_1) \\
& = & - \partial_x \left( A_1 \rho_1 u_1 \frac{1}{2} u_1^2 \right) + u_1 \left[ -A_1 \partial_x \mathcal{P}_2 - g \rho_1 A_1 \partial_x (\mathcal{B}+h)  \right]   \\
&& + \nu_1 A_1 \rho_1 u_1 \partial_x^2 u_1,
\end{array}
\]
\[
\begin{array}{lcl}
\partial_t (A_1 g (\mathcal{B}+h)) & = & g(\mathcal{B}+h) \partial_t A_1 + A_1 g \partial_t h = - g(\mathcal{B}+h) \partial_x (A_1 u_1) + \frac{g A_1}{\sigma_1} \partial_t A_1 \\
& = & - g(\mathcal{B}+h) \partial_x (A_1 u_1) - \partial_t (\rho_1 I_1).
\end{array}
\]
Adding the previous two equations, we get the entropy contribution coming from the bottom layer:
\begin{equation}
\label{eq:EntropyLowerLayer}
\begin{array}{ll}
\partial_t \left[ A_1 \rho_1 \left( \frac{1}{2}u_1^2 + g (\mathcal{B}+h)\right) - \rho_1 I_1 \right] & +  \partial_x \left[ A_1 \rho_1 u_1 \left( \frac{1}{2}u_1^2 +g (\mathcal{B}+h) \right) \right]  \\
 &  = -  A_1 u_1 \partial_x \mathcal{P}_2  +  \nu_1 u_1 \partial_x ( A_1 \rho_1 \partial_x u_1 ).
\end{array}
\end{equation}
As we can see, there is a term $A_1 u_1 \partial_x \mathcal{P}_2 $ that is not in conservation form, exchanging energy with the upper layer.

In order to get the total contribution towards entropy, we add the energy equation of the upper layer in \eqref{FEF1e}. The two terms inside the time derivative add to the entropy $\eta$. The two flux terms add to $q$, except for the term $\mathcal{P}_2 A_1 u_1$. That terms is obtained by adding the two non-conservative products in equations \eqref{FEF1e} and \eqref{eq:EntropyLowerLayer}. Those terms allow for the exchange of energy between layers while still conserving the overall contributions. We then obtain
\[
\partial_t \eta + \partial_x q = \nu_1 u_1 \partial_x ( A_1 \rho_1 \partial_x u_1 ).
\]
Since $\nu_1 > 0$, integration by parts show that the right hand side terms are negative in the weak sense. 

\end{proof}

%---------subsection-----------
\subsection{Hyperbolicity}

Hyperbolic conservation or balance laws constitute an important class of partial differential equations with a wide range of applications that explain various phenomena in gas dynamics and other areas. Such equations are of the form
\begin{equation}
\label{eq:UFS}
\textbf{U}_t + (\textbf{F}(\textbf{U}))_x =  \widetilde{\textbf{S}}.
\end{equation}
In our case, the vector of conserved variables, fluxes and source terms are given by
\[
\textbf{U} = 
\begin{pmatrix}
A_1\rho_1 \\
A_1 \rho_1 u_1\\
A_2 \rho_2 \\
A_2 \rho_2 u_2 \\ 
A_2 \rho_2 \mathcal{E}_2
\end{pmatrix},
\textbf{F}(\textbf{U},x)= 
\begin{pmatrix}
A_1 \rho_1 u_1 \\
A_1 (\rho_1 u_1^2+ \mathcal{P}_1) \\
A_2 \rho_2 u_2 \\
A_2 (\rho_2 u_2^2+ \mathcal{P}_2) \\
A_2 (\rho_2 \mathcal{E}_2+\mathcal{P}_2) u_2
\end{pmatrix},
\]
and
\begin{eqnarray}
\bf{\tilde S} = \begin{pmatrix}
    0\\
   \mathcal{P}_2 {\frac{\pr A_1}{\pr x}} - \sigma_{\mathcal{B}}\pr_x \mathcal{B} g \rho_1 h + g \rho_1\int_{\mathcal{B}}^{\mathcal{B}+h} (h +  \mathcal{B} -z)\pr_x\sigma dz \\
   0\\
   \mathcal{P}_2 \partial_x A_2  \\
   -\mathcal{P}_2 {\frac{\pr (A_1 u_1)}{\pr x}}
\end{pmatrix} .\nonumber
\end{eqnarray}
We note that the flux vector depends not only on $\vect U$ but also on the axial position $x$, due to the dependence of the pressure on the duct’s topography and geometry. Furthermore, the vector of source terms has been rewritten in a way that will be more convenient for us in the next section. In order for the conservation or balance law to be classified as hyperbolic, one first needs to write it in quasilinear form as
\ba
\textbf{U}_t + \textbf{M} \textbf{U}_x = \textbf{S}, \label{quasi}
\ea
where the coefficient matrix $\textbf{M}$ is given by
\ba
\begin{pmatrix}
0 & 1 & 0 & 0 & 0\\
A_1\left(\frac{\mathcal{P}_{2}}{A_2 \rho_1}+\frac{g}{\sigma_1} \right) -u_1^2& 2 u_1 & (\gamma_2 -1) \frac{1}{2}\frac{A_1u_2^2}{A_2} & -(\gamma_2 -1)\frac{A_1u_2}{A_2}  & (\gamma_2 -1)\frac{A_1}{A_2} \\
0 & 0 & 0 & 1 & 0 \\
\frac{\mathcal{P}_{2}}{\rho_1} & 0 & \frac{(\gamma_2 -3)}{2}u_2^2 & -u_2(\gamma_2 -3) & (\gamma_2-1)\\
0 & \frac{\mathcal{P}_{2}}{\rho_1} & (\gamma_2 -1)\frac{u_2^3}{2}-\frac{(\rho_2 \mathcal{E}_2 + \mathcal{P}_2)u_2}{\rho_2} & \frac{\rho_2 \mathcal{E}_2 + \mathcal{P}_2}{\rho_2}-(\gamma_2 -1)u_2^2 & \gamma_2 u_2
\end{pmatrix}. \label{quasilinearform}\hspace{0.5cm}
\ea
Here, we have assumed that $A_2 > 0$ is positive. That is, the duct is not totally full of water. Non-zero components of the vector of source terms write
\ba
S_{2} = \frac{A_1 \mathcal{P}_2}{A_2}\partial_x A_T+ \frac{-g\rho_1 A_1}{\sigma_1} \left(\sigma_ \mathcal{B} \pr_x \mathcal{B} - \int_{ \mathcal{B}}^{ \mathcal{B}+h} \pr_x\sigma dz \right), \\
S_{4} = \mathcal{P}_{2} \sigma_2 \pr_x H - \mathcal{P}_{2}\left(\sigma_ \mathcal{B} \pr_x \mathcal{B} - \int_{ \mathcal{B}}^{H} \pr_x\sigma dz \right).
\ea
We note that such terms do not contain any non-conservative products. 

We are interested in analyzing the hyperbolicity of the system. That is, one needs to analyze whether the eigenvalues of the coefficient matrix are all real with a complete set of eigenvectors.  The characteristic polynomial of the matrix \eqref{quasilinearform} in terms of the characteristic variable $\lambda$ is written as
\[
p(\lambda) = -(\lambda - u_2) \left[ \left( \left( \lambda- u_1 \right)^2 - c_1^2 \right) \left( \left( \lambda-u_2 \right)^2 -c_2^2 \right) -  \epsilon \; c_2^2 \left( \lambda-u_2 \right)^2 \right],
\]
where 
\[
c_1 = \sqrt{\frac{g A_1}{\sigma_1}}, \; c_2 = \sqrt{\frac{\gamma_2 P_2}{\rho_2}}, \text{ and } \epsilon = \frac{A_1 \rho_2}{A_2 \rho_1} 
\]
are the speed of sound for the bottom shallow water layer, the speed of sound of the top layer (typical of the Euler equations in gas dynamics),  and the non-dimensional ratio of densities and cross-sectional areas, respectively. 

As one can see, $\lambda = u_2$ is an eigenvalue, and it corresponds to the contact discontinuity (degenerate field) in the case of the Euler equations.  On the contrary, in general we do not have the facility of analytical expressions for the other 4 eigenvalues because of the last term in the characteristic polynomial. We note that in many situations, $\epsilon << 1$ will be a very small parameter. For instance, the ratio of densities between water and air satisfies $\rho_2/\rho_1 \approx 10^{-3}$. The parameter $\epsilon$ could also be small if $A_1/A_2 $ is small when the duct is almost dry, mainly occupied by gas. Of course, the opposite behavior when the duct is almost full of water could make $\epsilon$ large. 

In the limit when $\epsilon \to 0$, the 4 eigenvalues converge to $u_1 \pm c_1, \; u_2 \pm c_2$. In that limit, the two sub-systems are decoupled, two of those eigenvalues are associated with genuinely non-linear fields of the shallow water approximation (bottom layer) and the other two are also associated with genuinely non-linear fields of the Euler approximation (upper layer). 

In practice, our numerical algorithm includes the numerical computations of the eigenvalues, all of them real in the parameter regimes considered in our numerical tests. Nevertheless, the next theorem provides a theoretical analysis of the four eigenvalues that satisfy
\begin{equation}
\label{eq:4thDegreePoly}
\left( \left( \lambda- u_1 \right)^2 - c_1^2 \right) \left( \left( \lambda-u_2 \right)^2 -c_2^2 \right) = \epsilon \;  c_2^2 \left( \lambda-u_2 \right)^2 .
\end{equation}

\begin{theorem}
\label{thm:Hyp}
Assume that $A_1, A_2 > 0$ and $P_2 > 0$ are all positive. Assume also that 
\begin{equation}
\label{eq:DistintEvalues}
(u_2-u_1)^2 \neq (c_2 \pm c_1)^2, (u_2-u_1)^2 \neq c_1^2,  \text{ and } (u_2-u_1)^2 \neq c_2^2.
\end{equation}
Then, there exists a critical value $\epsilon_\text{crit} > 0$ that depends on $u_1,c_1,u_2,c_2$ such that system \eqref{FEF1} is hyperbolic for all $0 < \epsilon < \epsilon_\text{crit}$. 
\end{theorem}

We note that this theorem indicates that our system may lose hyperbolicity in the two extreme regimes, namely when the duct is almost completely filled with water or nearly dry and filled with gas. We also allow for the possibility that $\epsilon_{\text{crit}}$ may be equal to $+\infty$ in certain special cases, as will be shown below.

\begin{proof}

\begin{figure}[htp]
\centering
{\includegraphics[width=.46\textwidth]{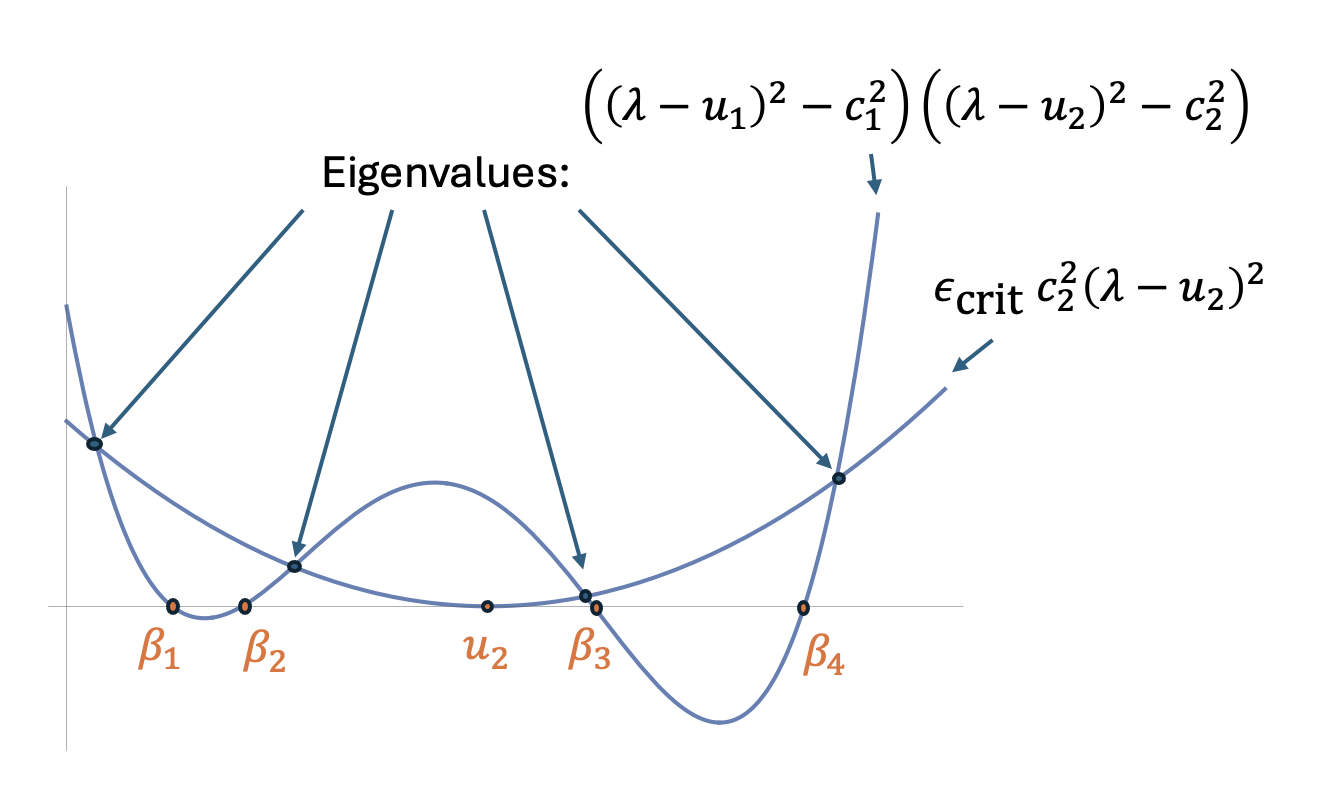}}
{\includegraphics[width=.49\textwidth]{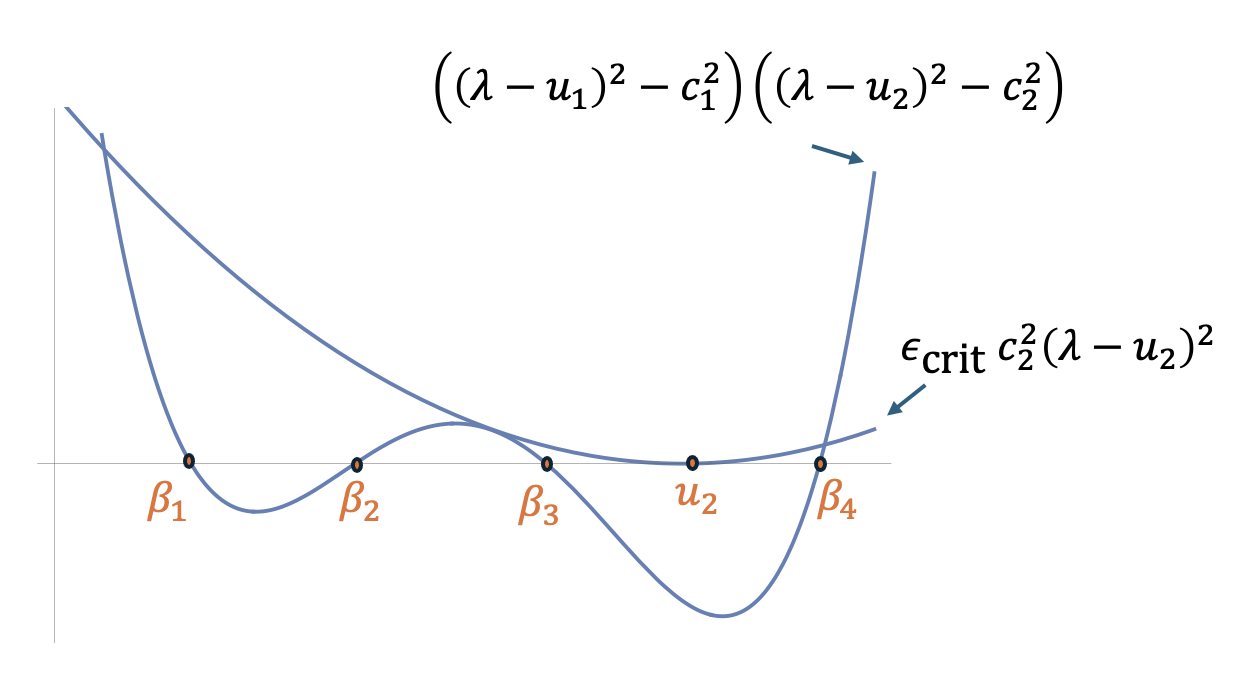}}
\caption{Schematic representation of the eigenvalue locations. Left: case of 4 distinct real eigenvalues. They are located at the intersections of the two curves defined by the polynomials on both sides of \eqref{eq:4thDegreePoly} (projected onto the horizontal axis). Right: case where one eigenvalue is a double zero, separating the regions of hyperbolicity. }
\label{fig:Eigenvalues}
\end{figure}

The left hand side of equation \eqref{eq:4thDegreePoly} is a quartic polynomial, with zeroes at $u_1 \pm c_1, \; u_2 \pm c_2$. Conditions in equation \eqref{eq:DistintEvalues} guarantees that those zeros are distinct. Let us arrange them in increasing order such that $\{ u_1 \pm c_1, u_2\pm c_2 \} = \{ \beta_1 < \beta_2 < \beta_3 < \beta_4 \}$. The quadratic polynomial on the right hand side is non-negative and vanishes at $u_2$. Due to the last conditions in \eqref{eq:DistintEvalues}, $u_2$ cannot be a zero of the quartic polynomial, and therefore  $ \beta_1 < u_2 < \beta_4$. 

Now, the quartic polynomial is positive in $[\beta_2,\beta_3]$. If $\beta_2 < u_2 < \beta_3$, the two polynomials cross in two points in that interval, and then two more crossing points exist, one in the interval $(-\infty,\beta_1)$ and the other one in $(\beta_4, \infty)$.  In that case, we have 5 distinct real eigenvalues (including $u_2$), obtaining a complete basis of eigenvectors. In this special case, the system is always hyperbolic, regardless of the value of $\epsilon$, leading to $\epsilon_{\text{crit}}= + \infty$. That situation is depicted in the left panel of Figure \ref{fig:Eigenvalues}. 

The other two possibilities are: $\beta_1 < u_2 < \beta_2$ or $\beta_3 < u_2 < \beta_4$. In any of these cases, for small enough $\epsilon$, the graph of the quadratic function in the right hand side of equation \eqref{eq:4thDegreePoly} is almost flat in the interval $[\beta_1,\beta_4]$, crossing the quartic polynomial in 4 distinct points. As we increase $\epsilon$, the quadratic polynomial becomes steeper and eventually intersects the quartic polynomial tangentially for a critical value $\epsilon_\text{crit}$. Larger values of $\epsilon$ will result in no crossing points in the interval $[\beta_2,\beta_3]$, giving two complex eigenvalues. This situation is depicted in the right panel of Figure \ref{fig:Eigenvalues}.
\end{proof}

\noindent
{\bf Upper and lower bounds for the eigenvalues.}
Since we have no explicit expressions for the eigenvalues, we would like to provide explicit bounds that could be useful for the computations of local speeds in the central-upwind scheme we use. For that end, we follow ideas of \cite{abgrall2009two}. 

\begin{theorem}
\label{thm:EvalueBounds}
Assume the hypotheses of Theorem \ref{thm:Hyp} and that $0 < \epsilon < \epsilon_{\text{crit}}$. Then, the eigenvalues are bounded below and above by
\[
\alpha_\text{min} = \min (\lambda_{1,*}^-,\lambda_{2,*}^-) \text{ and } \alpha_\text{max} = \max(\lambda_{1,*}^+, \lambda_{2,*}^+),
\]
where
\[
\lambda_{1,*}^\pm = u_1 \pm \sqrt{c_1^2 + (1+\epsilon) c_2^2} , \text{ and } \lambda_{2,*}^\pm = u_2 \pm \sqrt{1+\epsilon} \; c_2 .
\]
\end{theorem}

\begin{proof}
The bounds defined above satisfy
\[
(\lambda_{1,*}^\pm-u_1)^2 -c_1^2 = r_1 := (1+\epsilon)c_2^2, \; (\lambda_{2,*}^\pm-u_2)^2-c_2^2 = r_2 := \epsilon c_2^2. 
\]
Suppose for instance that $\alpha_\text{max} = \lambda_{1,*}^+$. Then
\[
\begin{array}{lcl}
q(\alpha_\text{max}) & = & ((\alpha_\text{max}-u_1)^2-c_1^2)((\alpha_\text{max}-u_2)^2-c_2^2)- \epsilon c_2^2 (\alpha_\text{max}-u_2)^2 \\
& = &  r_1 ((\lambda_{1,*}^\pm-u_2)^2-c_2^2)- \epsilon c_2^2 (\lambda_{1,*}-u_2)^2 \\
& = & (r_1-\epsilon c_2^2) (\lambda_{1,*}^+-u_2)^2-r_1 c_2^2 \\
& \ge & (r_1-\epsilon c_2^2) (c_2^2+r_2) - r_1 c_2^2 \text{ if }  r_1-\epsilon c_2^2 \ge 0 \text{ because } \lambda_{1,*}^+ \ge \lambda_{2,*}^+ \\
& = & (r_1 - \epsilon c_2^2) r_2 - \epsilon c_2^4 .
\end{array}
\]

The other option is that $\alpha_\text{max} = \lambda_{2,*}^+$. In that case,
\[
\begin{array}{lcl}
q(\alpha_\text{max}) & = &  ((\lambda_{2,*}^\pm-u_1)^2-c_1^2) r_2 - \epsilon c_2^2 (\lambda_{2,*}-u_2)^2 \\
& = &  ((\lambda_{2,*}^\pm-u_1)^2-c_1^2) r_2 - \epsilon c_2^2 (c_2^2+r_2) \\
& \ge & r_1 r_2 - \epsilon c_2^2 (c_2^2+r_2) \text{ because } \lambda_{2,*}^+ \ge \lambda_{1,*}^+
= (r_1 - \epsilon c_2^2) r_2 - \epsilon c_2^4 .
\end{array}
\]
We note that $(r_1-\epsilon c_2^2) r_2 - \epsilon c_2^4 = 0$ if $r_1$ and $r_2$ are defined as above. This shows that  $q(\alpha_\text{max})  \ge 0$, and similarly $q(\alpha_\text{min}) \ge 0$. Now, but construction, $\alpha_\text{min} \le \beta_1$ and $\alpha_\text{max} \ge \beta_4$. So, $\alpha_\text{min}$ has to stay to the left of the first eigenvalue and $\alpha_\text{max}$ must be greater than the largest one. 

\end{proof}

\noindent
{\bf Asymptotic approximations of the eigenvalues as $\epsilon \to 0$.}
The factor $\epsilon$ is often small if the gas density is much smaller compared to the density of the lower layer in liquid phase, or if $A_1/A_2$ is small. For instance, the ratio of densities between air and water is about $10^{-3}$. The following theorem gives us approximated eigenvalues for $\epsilon << 1$ small.

\begin{theorem}
\label{thm:EvalueApprox}
Assume that $(u_2-u_1)^2 \neq (c_2 \pm c_1)^2$ and that $0 < \epsilon < \epsilon_{\text{crit}}$. The eigenvalues of the coefficient matrix \eqref{quasilinearform} satisfy the following asymptotic approximations as $\epsilon \to 0$, maintaining all other parameters fixed:
\[
\begin{array}{lclcl}
\lambda_1^- & = & \gamma_1^- + O(\epsilon^2), & \text{where} & \gamma_1^- =  u_1 - c_1 -  \frac{c_2^2}{2 c_1} \frac{(u_1- c_1-u_2)^2}{(u_1- c_1-u_2)^2-c_2^2}  \epsilon, \\
\lambda_1^+ & = & \gamma_1^+ + O(\epsilon^2), & \text{where} & \gamma_1^+ = u_1 + c_1 +  \frac{c_2^2}{2 c_1} \frac{(u_1+ c_1-u_2)^2}{(u_1+ c_1-u_2)^2-c_2^2}  \epsilon, \\
\lambda_{2}^- & = & \gamma_1^- + O(\epsilon^2), & \text{where} & \gamma_1^- = u_2 - c_2 - \frac{c_2}{2} \frac{c_2^2}{(u_2- c_2-u_1)^2-c_1^2}  \epsilon,\\
\lambda_{2}^0 & = &  u_2, \text{ and }&& \\
\lambda_{2}^+ & = & \gamma_1^- + O(\epsilon^2), & \text{where} & \gamma_1^- =  u_2 + c_2 \pm \frac{c_2}{2} \frac{c_2^2}{(u_2\pm c_2-u_1)^2-c_1^2}  \epsilon.
\end{array}
\]
\end{theorem}

\begin{proof}
Substituting an expansion $\lambda = \lambda_o + \lambda_1 \epsilon + O(\epsilon^2)$, and collecting terms of the same order, we get
\[
((\lambda_0-u_1)^2-c_1^2 ) ((\lambda_0-u_2)^2-c_2^2) = 0, \text{ and }
\]
\[
(\lambda_0-u_1) \lambda_1 ((\lambda_0-u_2)^2-c_2^2) + (\lambda_0-u_2) ((\lambda_0-u_1)^2-c_1^2) - \frac{c_2^2}{2} (\lambda_0-u_2)^2 = 0.
\]
We get to the conclusions after solving the above equations for the different cases. 
\end{proof}

\begin{figure}[htp]
\centering
{\includegraphics[width=.8\textwidth]{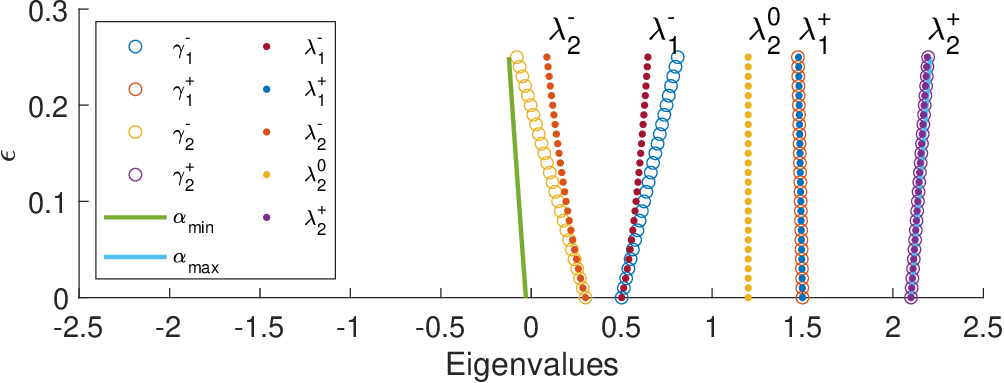}}
\caption{Plot of eigenvalues versus $\epsilon$, maintaining other parameters fixed: $u_1 = 1,\; c_1 = 0.5, \;, u_2 = 1.2, \; c_2 = 0.9$. Approximations are also displayed ($\gamma_{k}^\pm, k = 1,2$), together with the bounds $\alpha_\text{min}, \; \alpha_\text{max}$. }
\label{fig:EigenvaluesApproxBounds}
\end{figure}

We test the theorems in this section in the plot shown in Figure \ref{fig:EigenvaluesApproxBounds}.  We display the exact eigenvalues (horizontal axis) for different values of $\epsilon$ (vertical axis). Also plotted are the eigenvalue approximations $\gamma_{k}^\pm, \; k =1,2$ obtained in Theorem \ref{thm:EvalueApprox}. As one can see, for the chosen fixed parameters ($u_1 = 1,\; c_1 = 0.5, \;, u_2 = 1.2, \; c_2 = 0.9$), excellent eigenvalue approximations are observed for $\lambda_{1}^+$ and $\lambda_{2}^+$ for all values of $\epsilon \in [0, 0.25]$. In contrast, the eigenvalue approximations for $\lambda_1^-$ and $\lambda_2^-$ are good only for small values of $\epsilon$. The green and blue solid lines indicate the bounds obtained according to Theorem \ref{thm:EvalueBounds}. For the chosen parameter regime, the upper bound seems to be quite optimal. Although the lower bound is somewhat smaller than the first eigenvalue, it is of the same order of magnitude. Of course, how optimal the bounds and approximations are may depend on the parameter regime.

%%-----
\section{Numerical Results} 
\label{sec:NumResults}

In this section, we present several numerical tests that show the merits of the model and the numerical algorithm. We use a well-balanced, positivity-preserving central-upwind scheme for approximating solutions to the formulated novel model. For further background on central and central-upwind methods, we refer the reader to \cite{kurganov2009central, kurganov2000new}.

Most of the tests involve a situation where the gas density is much smaller than the liquid density. This occurs for instance, in ducts with water and air. In those cases, one would expect that any perturbation of the liquid layer directly affects the evolution of the gas above it. On the contrary, the influence of the gas layer on the liquid through momentum and energy exchanges between layers is expected to be very weak. The feedback is then mostly one way. Cases of interest may also include liquid and gas interactions where the density difference is not so pronounced, like liquid and gas hydrogen. We include one such numerical test where the feedback between the two layers is significant in both directions.

In all numerical tests, we implement free boundary conditions. That is, we impose Dirichlet boundary conditions (coinciding with the initial condition) at inflow, and apply extrapolation at outflow. An inflow/outflow occurs when the coefficient matrix has eigenvalues that point inward/outward. Unless otherwise stated, we use the following model parameter values: $g = 9.8$ and $\gamma_2 = 1.4$. 

\subsection{Cylindrical pipe - well-balance test}

\begin{figure}
\centering
\begin{subfigure}{0.49\textwidth}
    {\includegraphics[width=.99\textwidth]{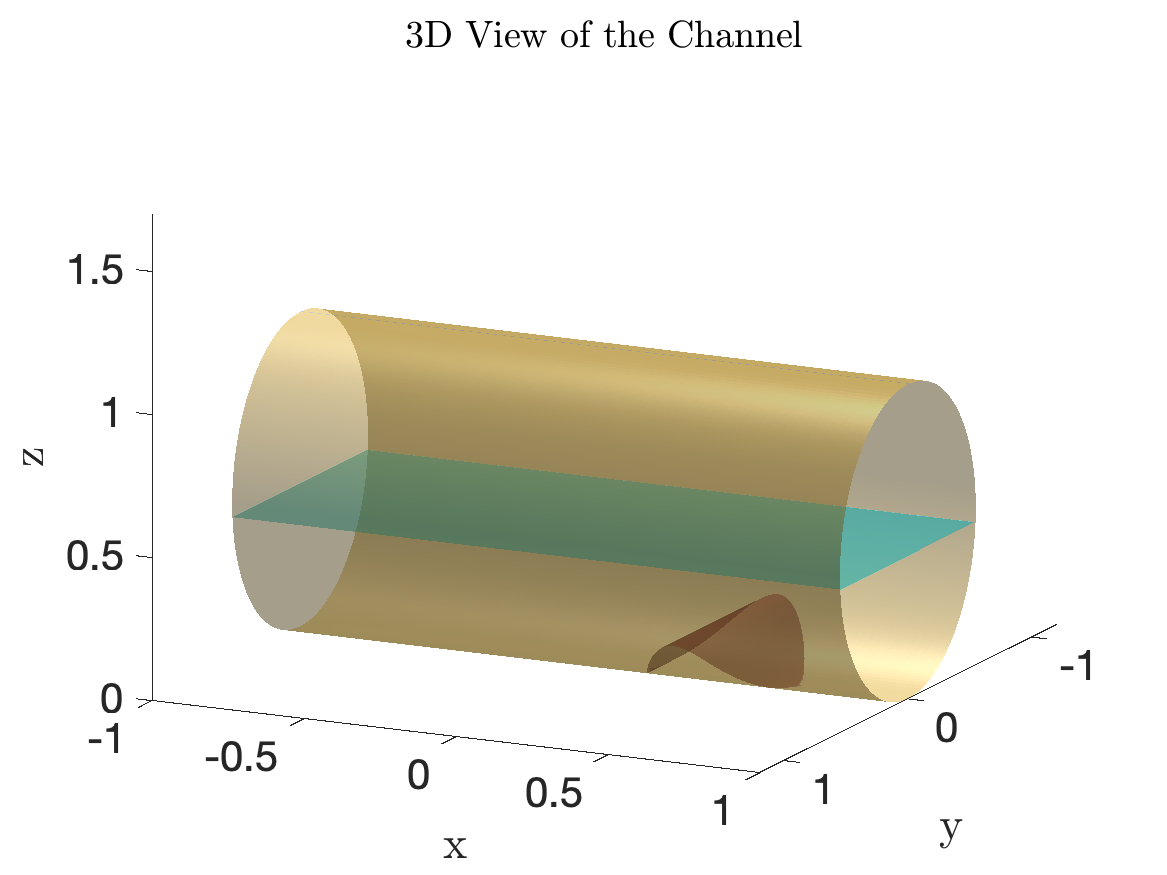}}
    \caption{3D view.}
\end{subfigure}
\begin{subfigure}{0.49\textwidth}
	{\includegraphics[width = 0.49 \textwidth]{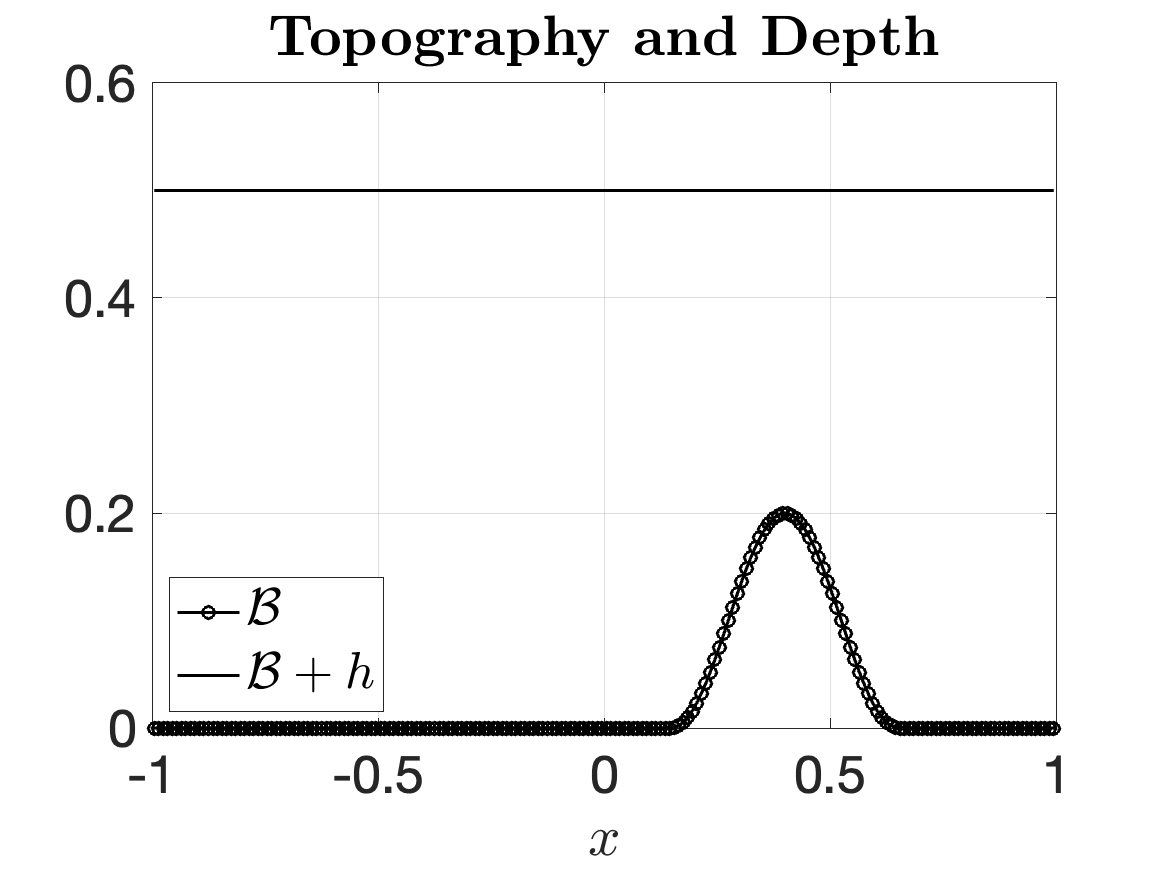}}
	{\includegraphics[width = 0.49\textwidth]{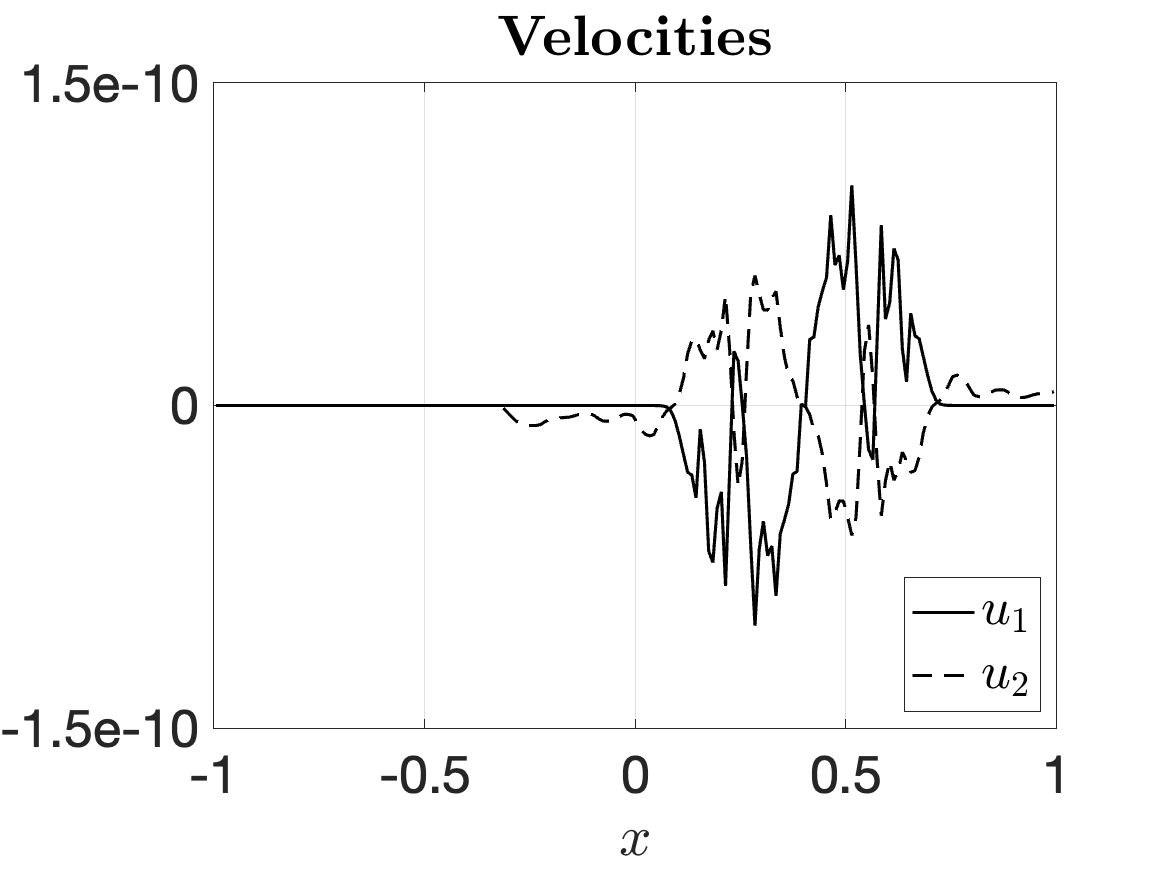}}\\
	{\includegraphics[width = 0.49 \textwidth]{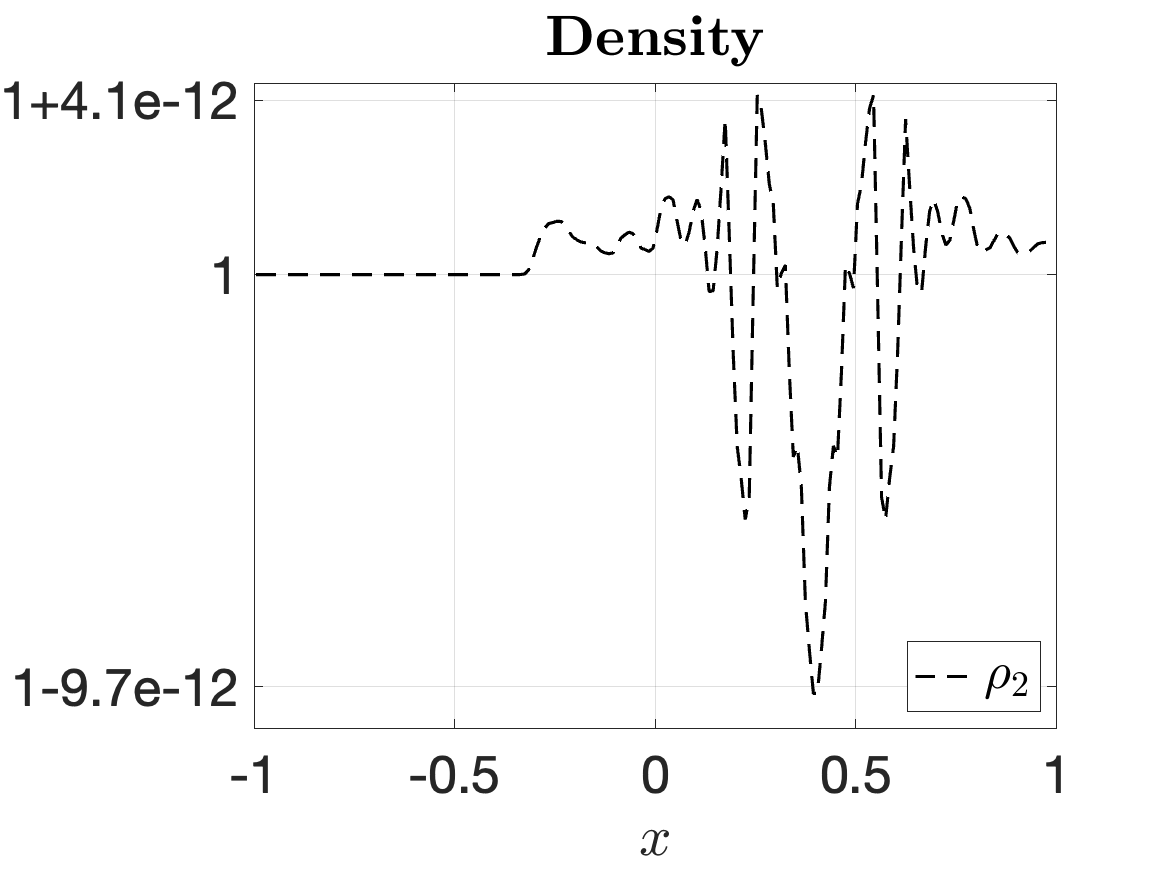}}
	{\includegraphics[width = 0.49 \textwidth]{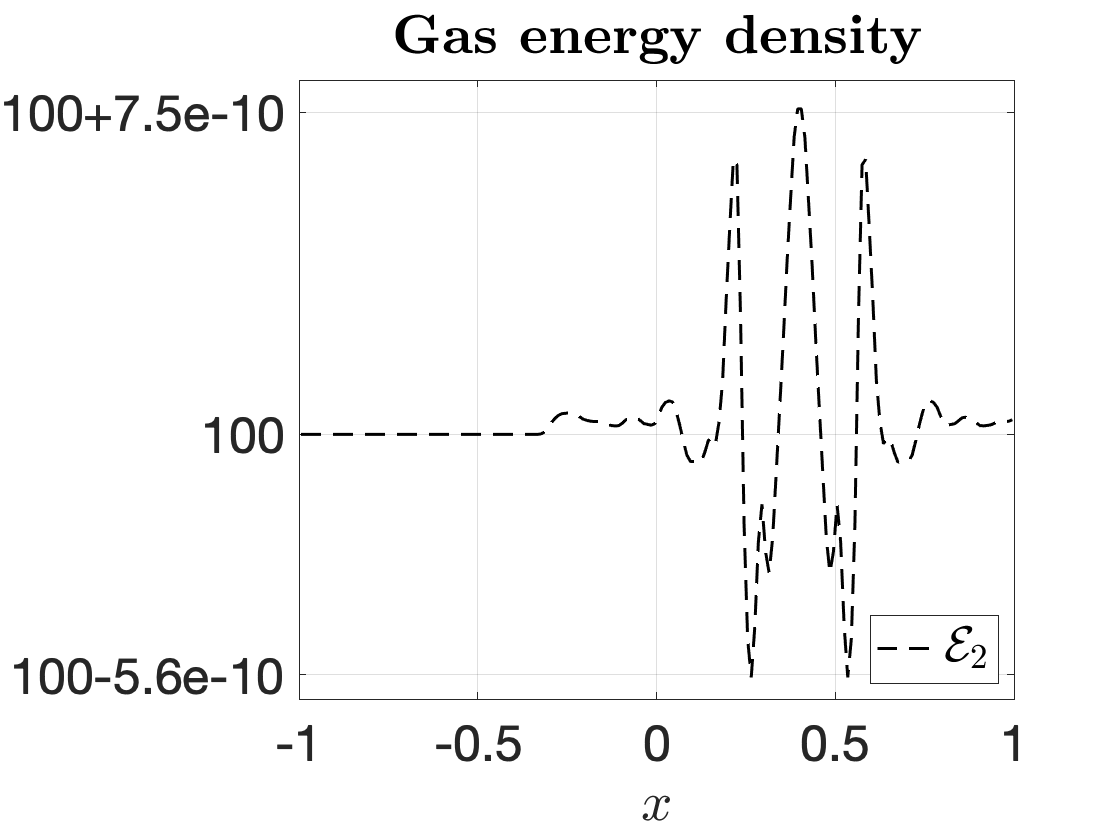}}
    \caption{Profiles.}
\end{subfigure}        
\caption{Left: 3D view of the duct at time $t= 0.06$ with initial conditions, topography and channel's width given by \eqref{eq:IC_wellbalance}, \eqref{bottombump} and \eqref{cylinder} respectively. Right: Profiles of topography and interface (top left), velocities (top right), gas density (bottom left) and gas energy density (bottom right) at  $t=0.06$ using $200$ grid points. }
\label{fig:well-balance}
\end{figure}

In this numerical test, we first consider a continuous topography, which is given by
\begin{eqnarray} \label{bottombump}
\mathcal{B}(x) =
\begin{cases}
\frac{1}{10} \left(\cos\left(\frac{\pi(x-0.4)}{0.25}\right) + 1\right) \hspace{1.0cm} \mbox{if}\,\, x \in [0.15,0.65], \text{ and }\\
0 \hspace{4.5cm} \mbox{if}\,\, x \in [-1,1]\setminus [0.15,0.65].
\end{cases} 
\end{eqnarray}

The channel's wall is described by the width $\sigma(x,z)$ at each axial position $x$, and height $z$, given by
\begin{eqnarray} \label{cylinder}
\sigma(x,z) = 2 \sqrt{\max\left(0,\ r_{0}^2 - (r_{0} - z)^2\right)}, \quad \text{with } r_{0} = \frac{1.1}{2}.
\end{eqnarray}
This corresponds to a cylindrical pipe.
%The parameter values are  $g = 9.8$, and $\gamma_2 = 1.4$ (for air). 
The initial conditions are given by
\begin{eqnarray}
\label{eq:IC_wellbalance}
\rho_1 =1000 ,\,\,\, \rho_2(x,0) = 1, \,\,\,\, u_1(x,0) = u_2(x,0) = 0,\nonumber\\
\,\,\, \mathcal{B}(x) + h(x,0) = 0.5, \,\,\, \mathcal{E}_2(x,0)  = 100. 
\end{eqnarray}

The above initial conditions correspond to a steady state of rest. Our numerical scheme is well balanced, respecting such initial conditions. We confirm this important property with the results shown in Figure \ref{fig:well-balance}. The left panel shows the 3D view of the channel, indicating the topography in brown color at the bottom, while the interface is distinguished by the blue surface. The right panel shows profiles of different quantities. The top left position shows the topography and the (flat) interface. The two velocities are shown in the top right position, showing numerical errors of order $10^{-10}$. The bottom left position shows the density of the gas layer with numerical errors of order $10^{-12}$, while the energy in the bottom right position displays numerical errors of order $10^{-10}$. As one can see, the numerical scheme respects this steady state of rest with numerical perturbations associated only to round-off errors.

%\newpage

\subsection{Riemann problem}

\begin{figure}
\centering
\begin{subfigure}{0.47\textwidth}
    {\includegraphics[width=.99\textwidth]{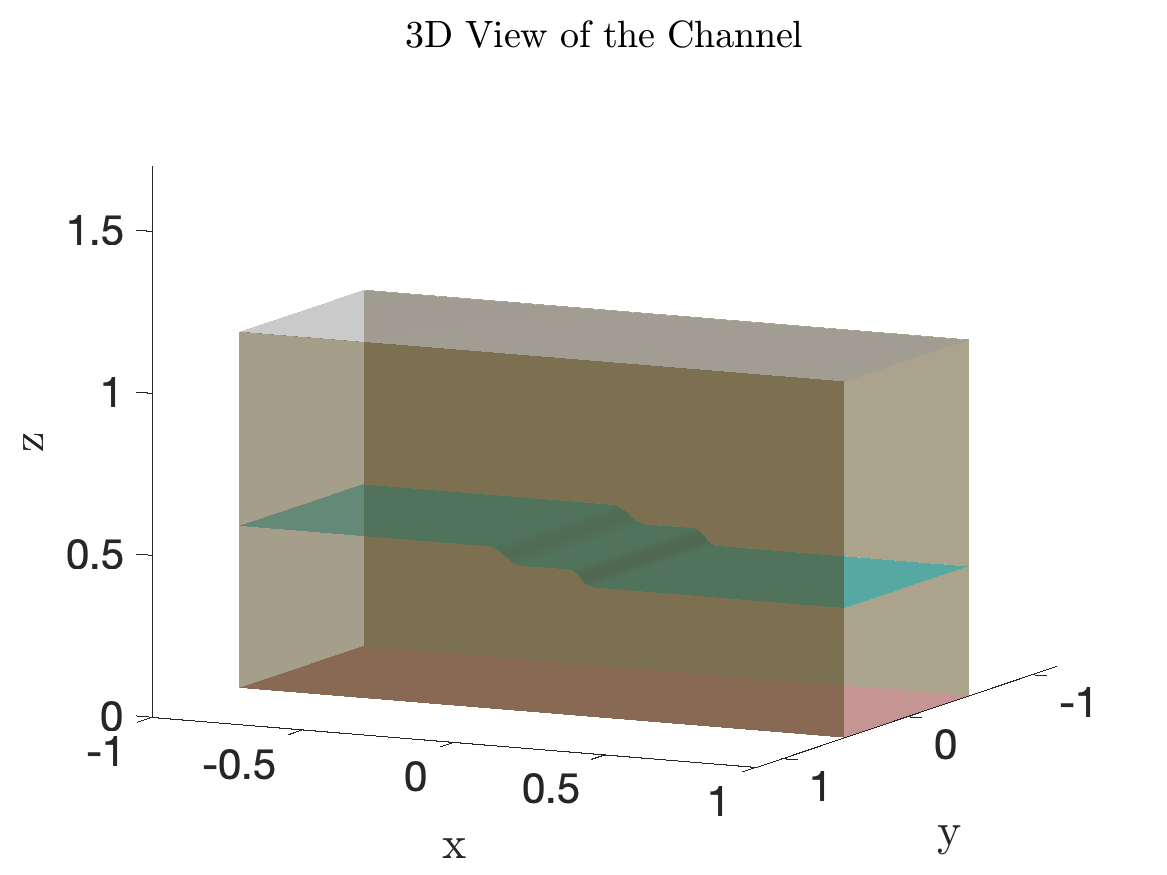}}
    \caption{3D view.}
\end{subfigure}
\begin{subfigure}{0.45\textwidth}
	{\includegraphics[width = 0.46 \textwidth]{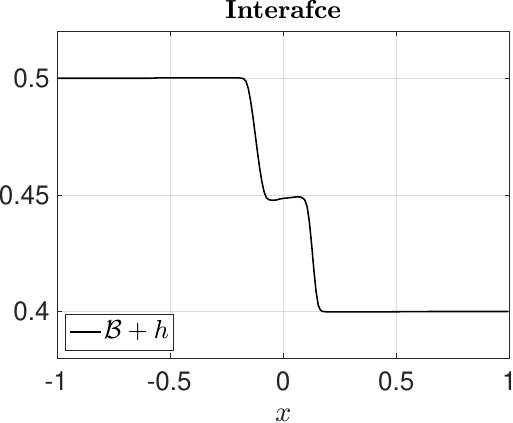}}
	{\includegraphics[width = 0.45\textwidth]{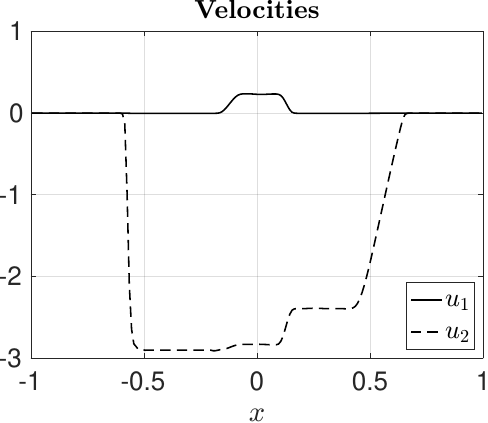}}\\
	{\includegraphics[width = 0.45\textwidth]{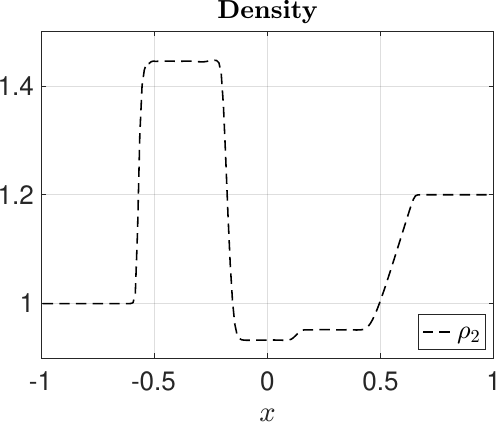}}
	{\includegraphics[width = 0.45\textwidth]{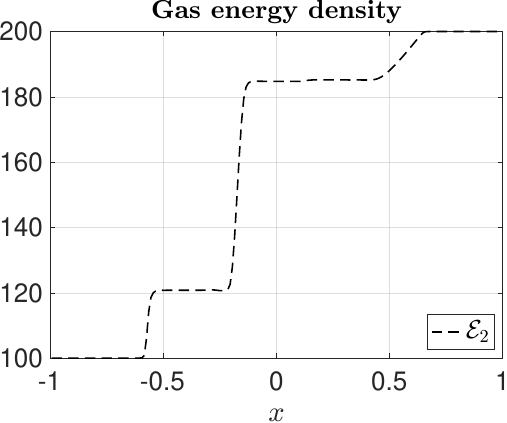}}
    \caption{Profiles.}
\end{subfigure}        
\caption{Solution at $t=0.06$  using $200$ grid points. The initial conditions are given by \eqref{eq:IC_RP}. Left panel: 3D view of the channel. Right panel: interface (top left), velocities (top right), gas density (bottom left) and gas energy density (bottom right) are shown. }
\label{fig:riemann}
\end{figure}

%The model's parameter values are  $g = 9.8$, and $\gamma_2 = 1.4$.
In this Riemann problem test, the bottom topography is flat ($\mathcal{B} = 0$) and the channel has constant width ($\sigma = 1$).   The initial conditions consist of piecewise constant functions with discontinuities only in $h, \mathcal{E}_2$, and $\rho_2$:
\begin{eqnarray}
\label{eq:IC_RP}
\rho_1 =1000 ,\,\, u_1(x,0) = u_2(x,0) = 0, \,\, h(x,0)  = \begin{cases}
0.5 \hspace{0.5 cm} \text{ if }  -1 \le x < 0, \\
0.4  \hspace{0.5cm} \text{ if }   0 \le x \le 1 , 
\end{cases} \nonumber\\
\rho_2(x,0)  = \begin{cases}
1 \hspace{0.5 cm} \text{ if }  -1 \le x < 0, \\
1.2  \hspace{0.5cm} \text{ if }   0 \le x \le 1.
\end{cases},\, \mathcal{E}_2(x,0) = \begin{cases}
100 \hspace{.5cm} \text{ if }   -1 \le x < 0, \\
200  \hspace{.5cm} \text{ if }   0 \le x \le 1.
\end{cases}
\end{eqnarray}

Figure \ref{fig:riemann} shows the solution at $t = 0.06$.  The gas energy density plot (bottom right panel) clearly identifies one left shockwave, one contact discontinuity and one right rarefaction wave. Other variations in the gas density (bottom left panel)  arise due to the influence of the waves associated with the lower layer. The waves associated with the lower layer are clearly identified with the variations in the interface (top left panel) and the liquid velocity (top right panel), one of them located near the gas contact discontinuity. As a result, small variations in the gas velocity are also observed there.

%\newpage

\subsection{Perturbation (in water layer) and convergence to a steady state of rest}

\begin{figure}[htp]
\centering
{\includegraphics[width=.4\textwidth]{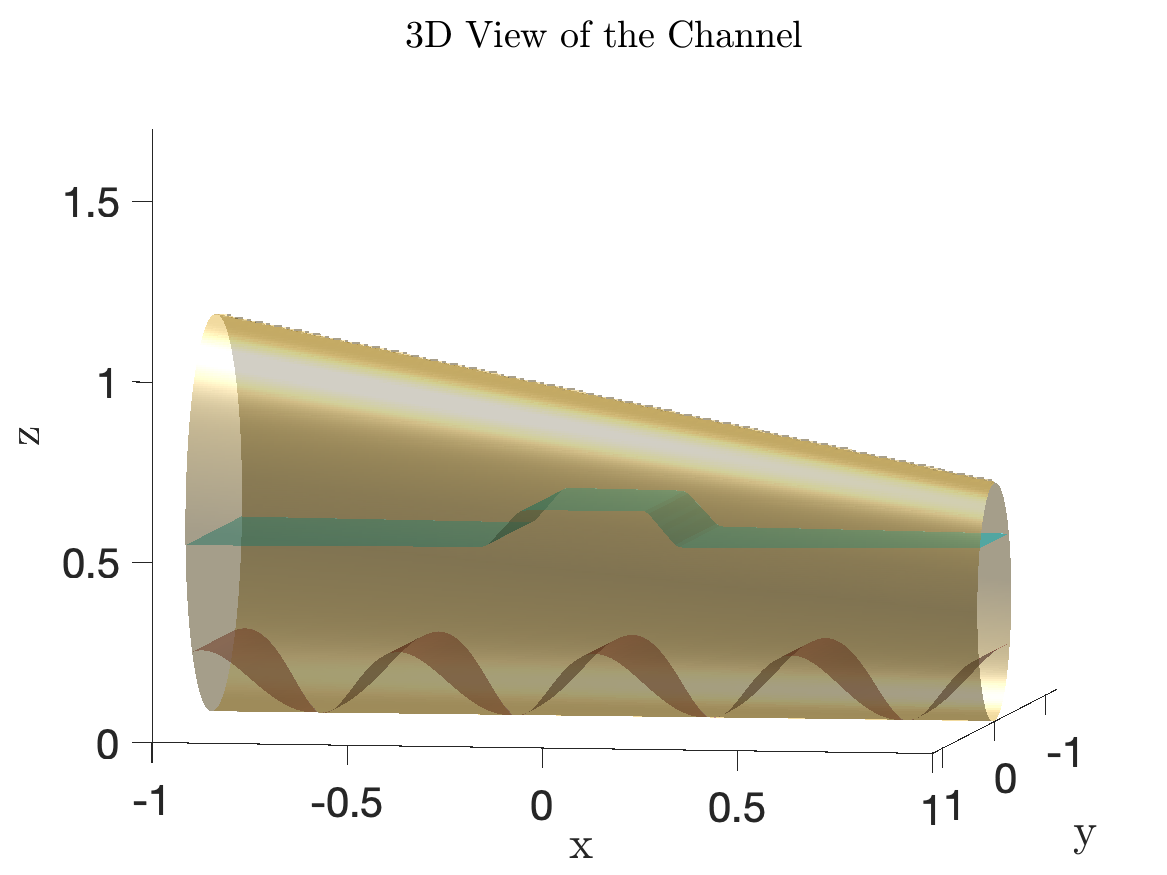}}
{\includegraphics[width=.4\textwidth]{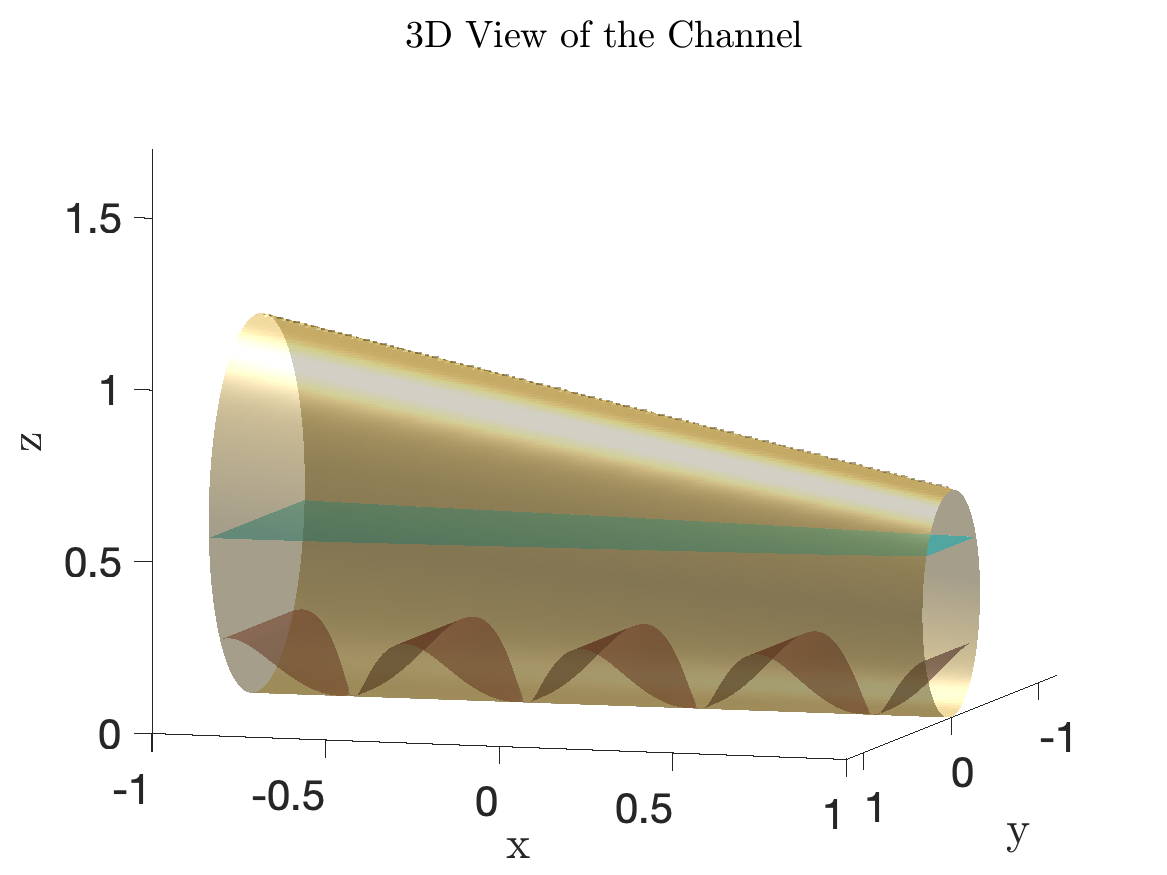}}\\
{\includegraphics[width = 0.21 \textwidth]{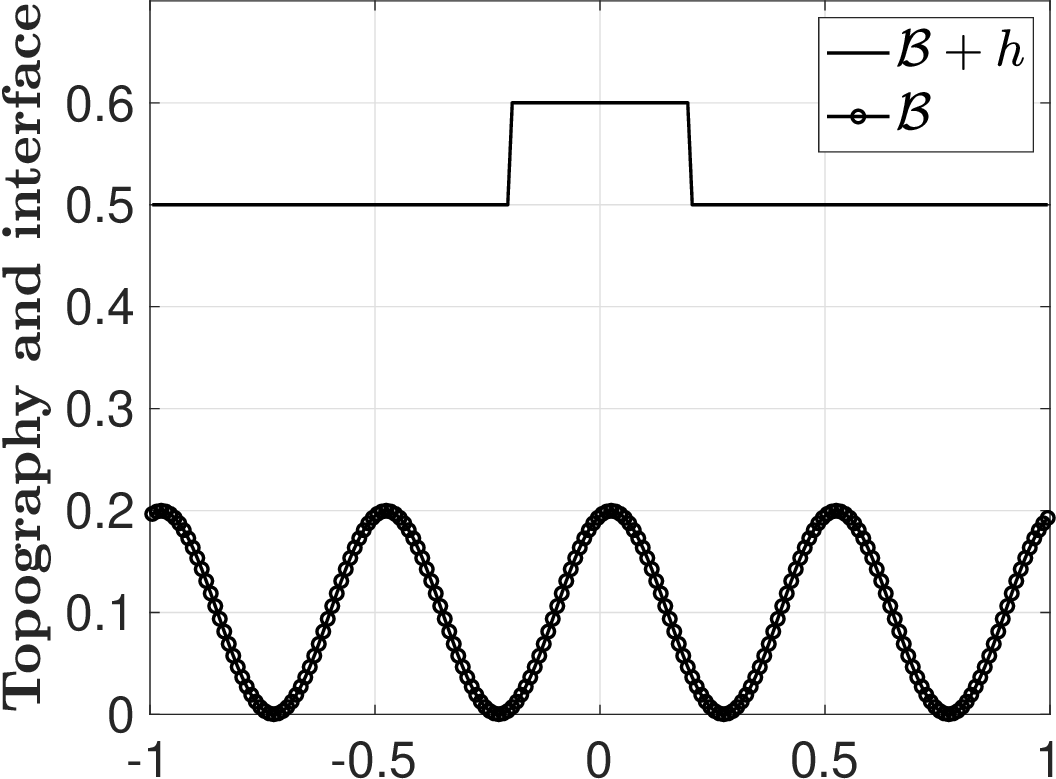}}
{\includegraphics[width = 0.2\textwidth]{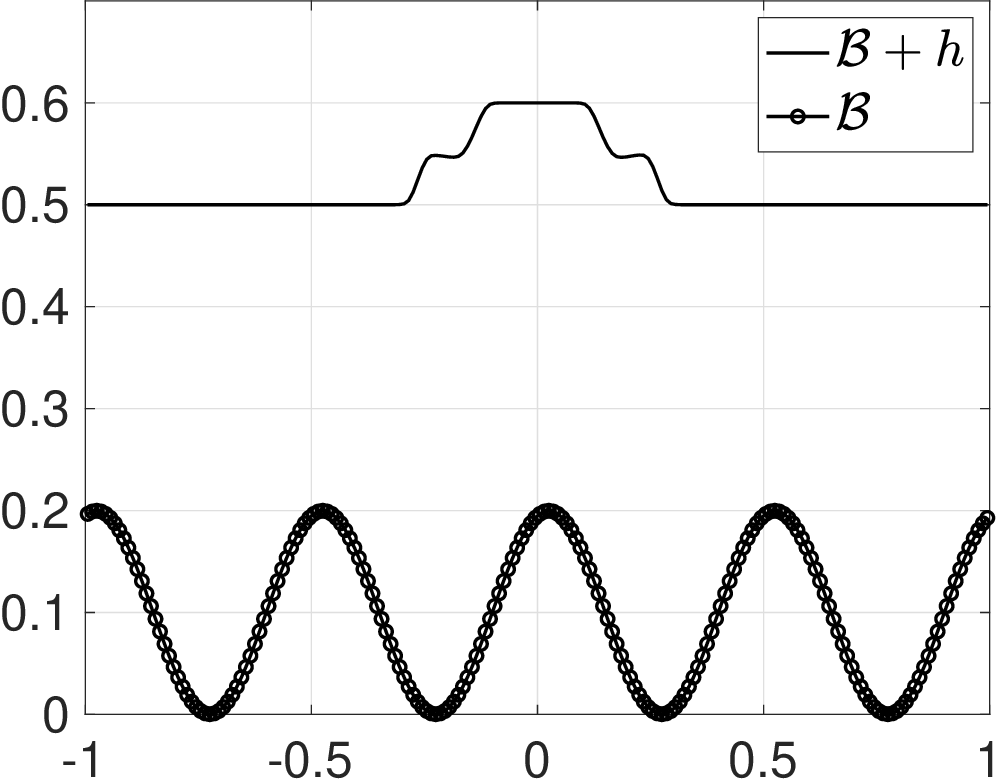}}
{\includegraphics[width = 0.2\textwidth]{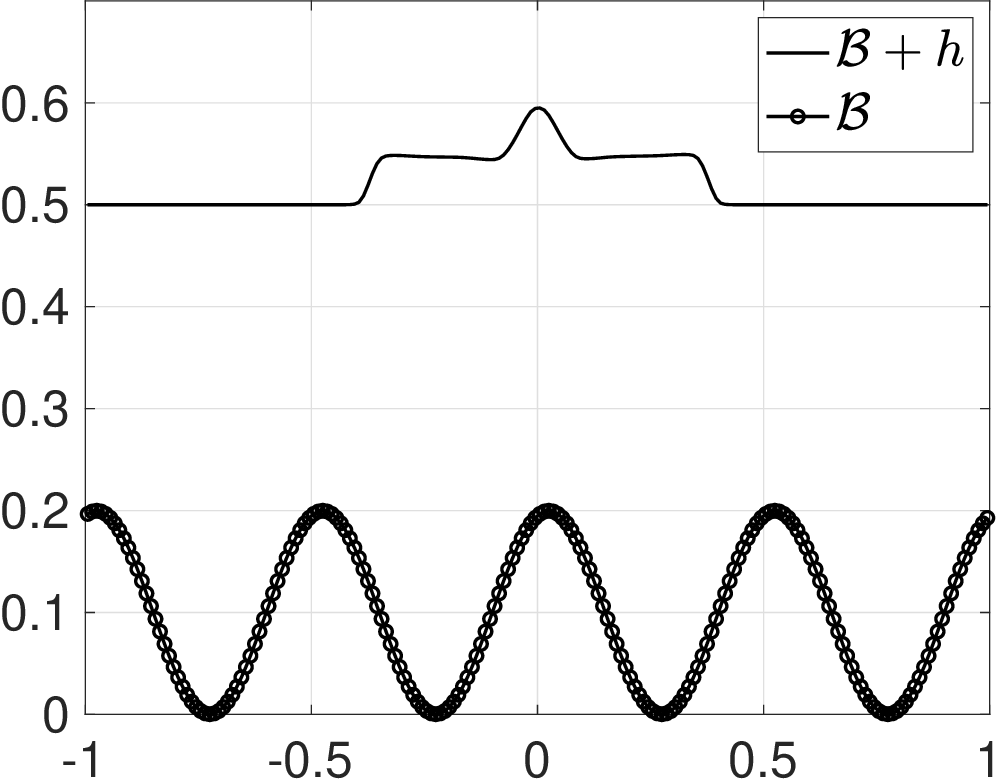}}
{\includegraphics[width = 0.2 \textwidth]{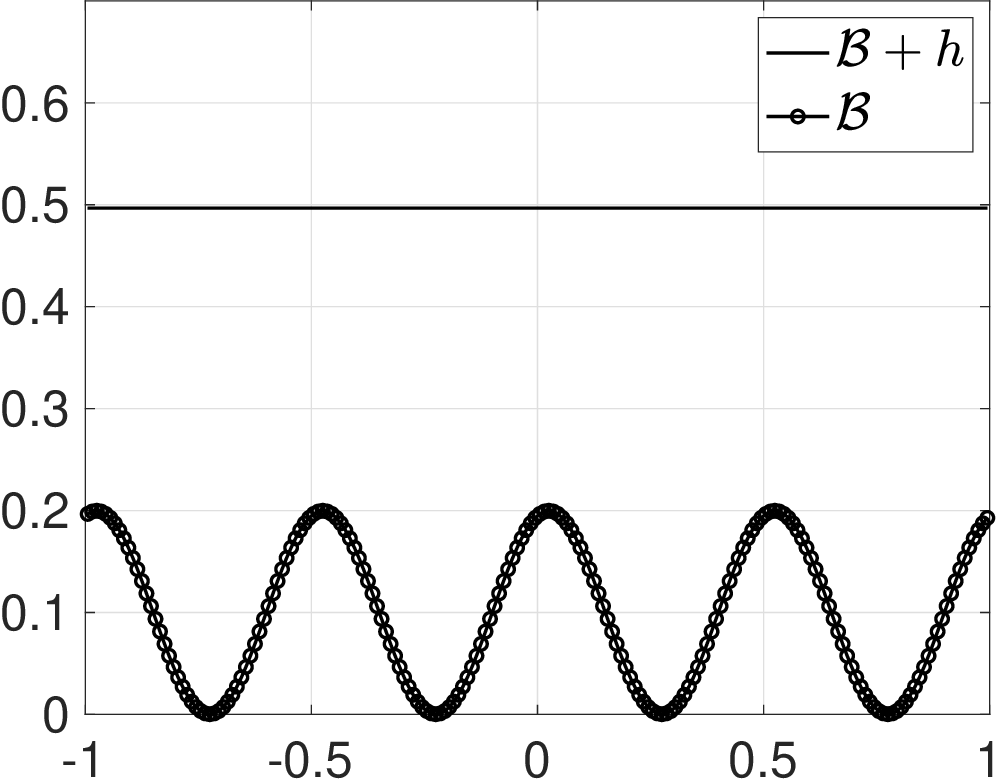}}\\
{\includegraphics[width = 0.21 \textwidth]{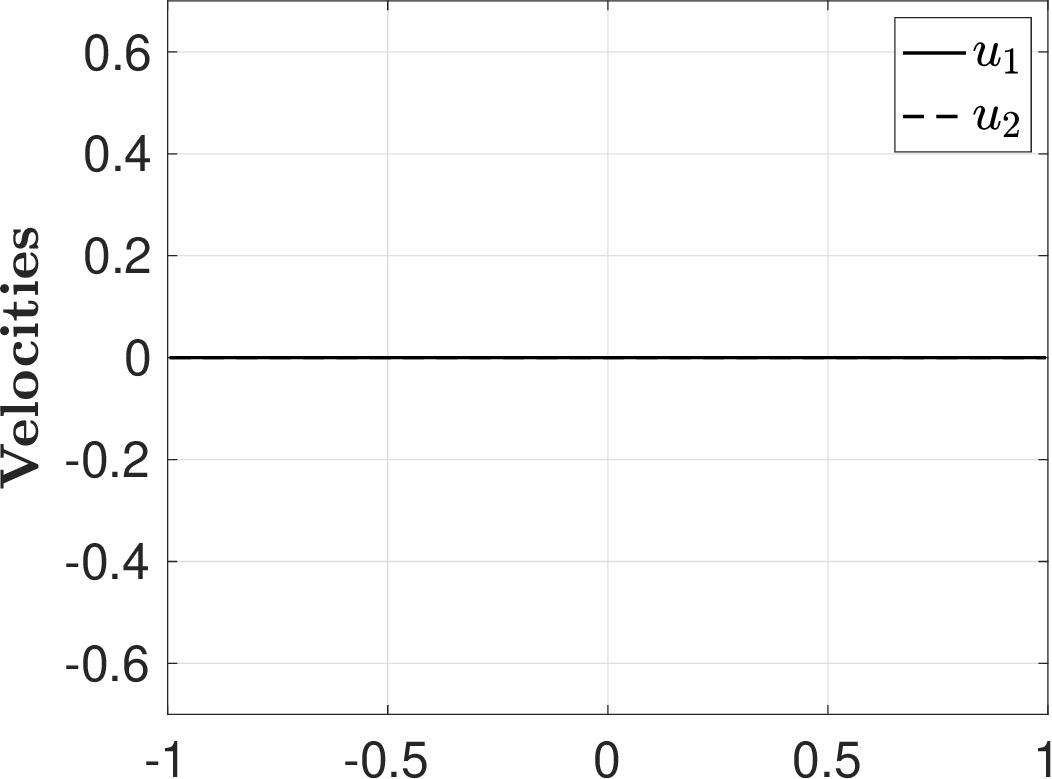}}
{\includegraphics[width = 0.2 \textwidth]{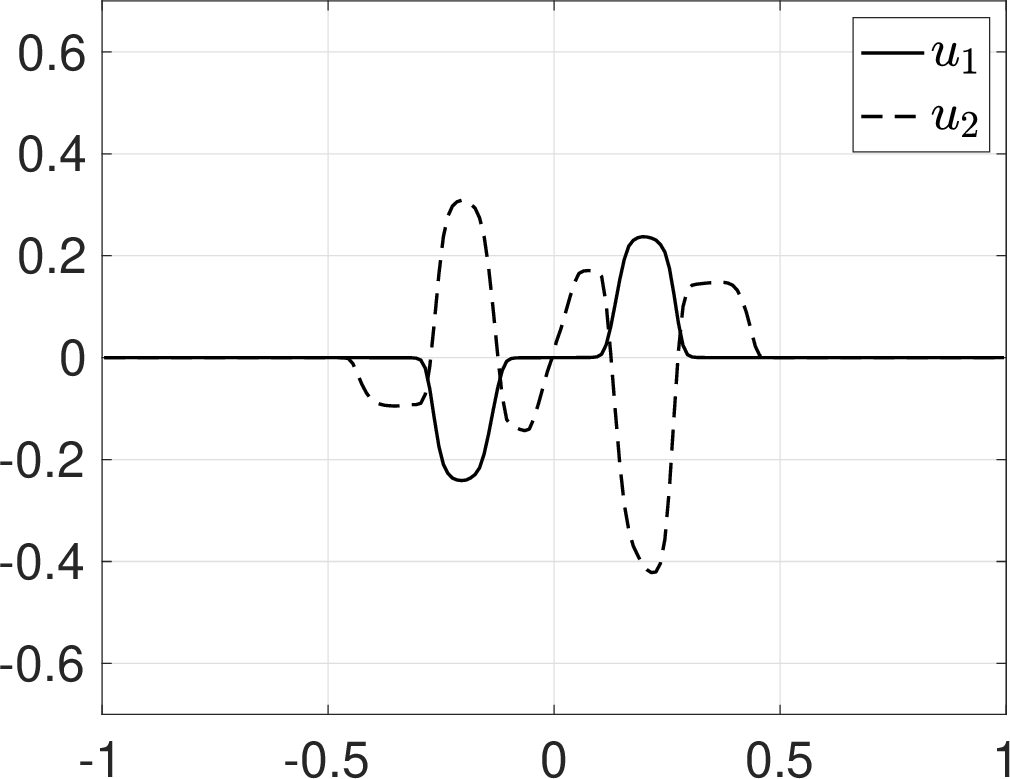}}
{\includegraphics[width = 0.2 \textwidth]{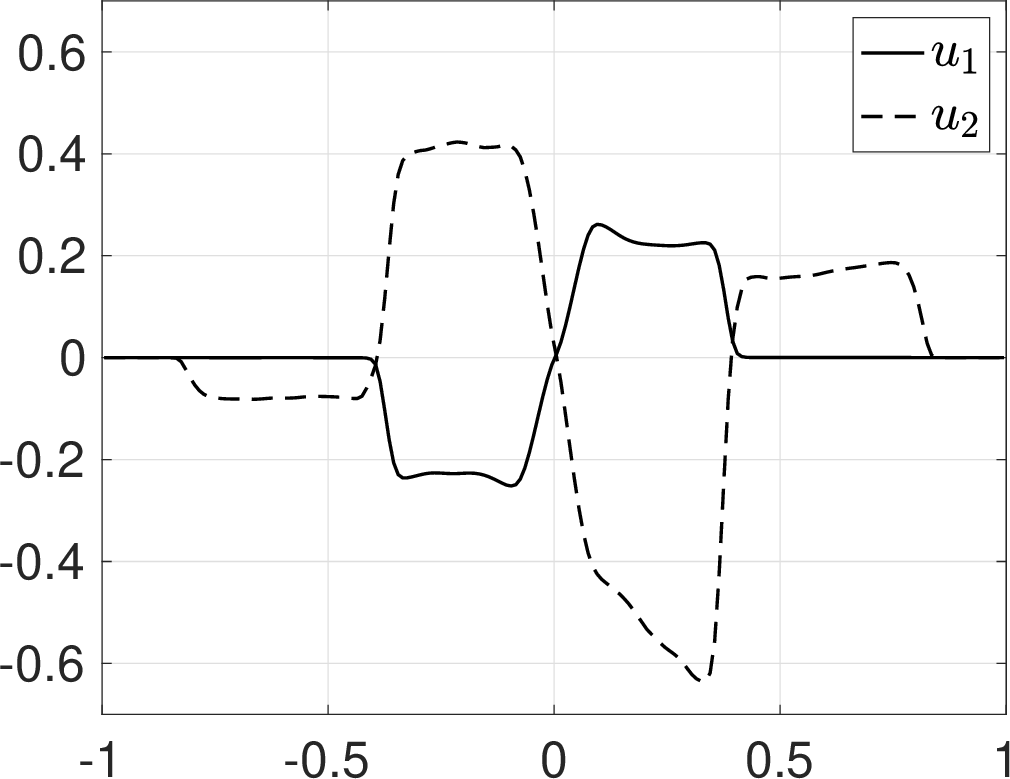}}
{\includegraphics[width = 0.2 \textwidth]{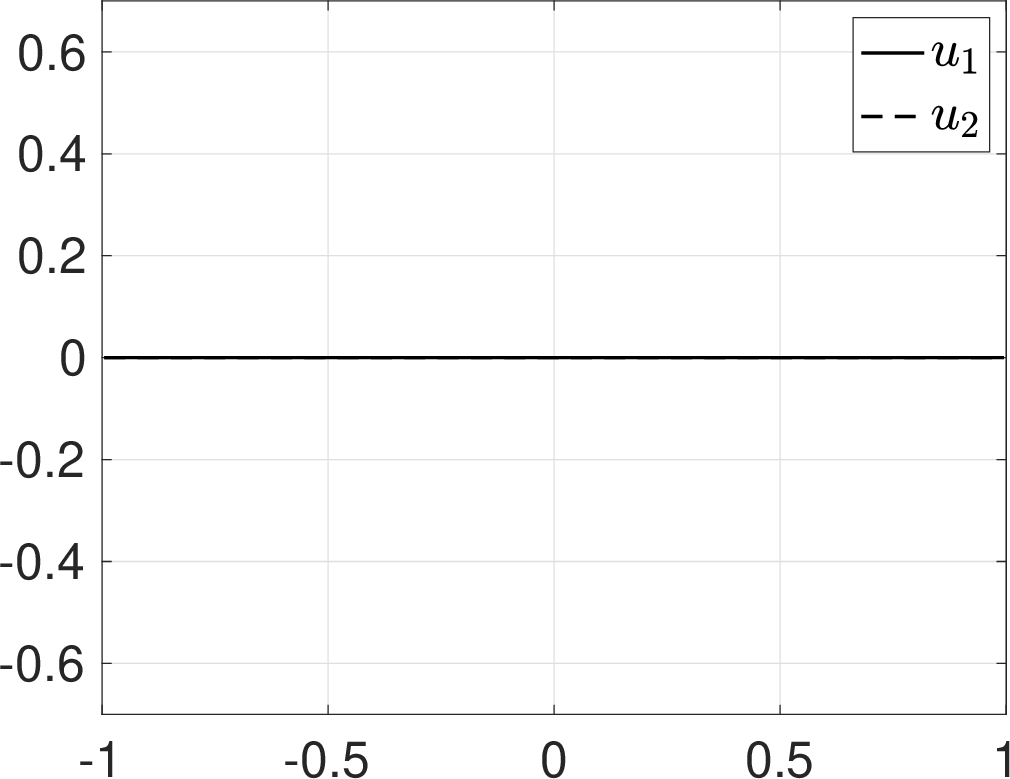}}\\
{\includegraphics[width = 0.21 \textwidth]{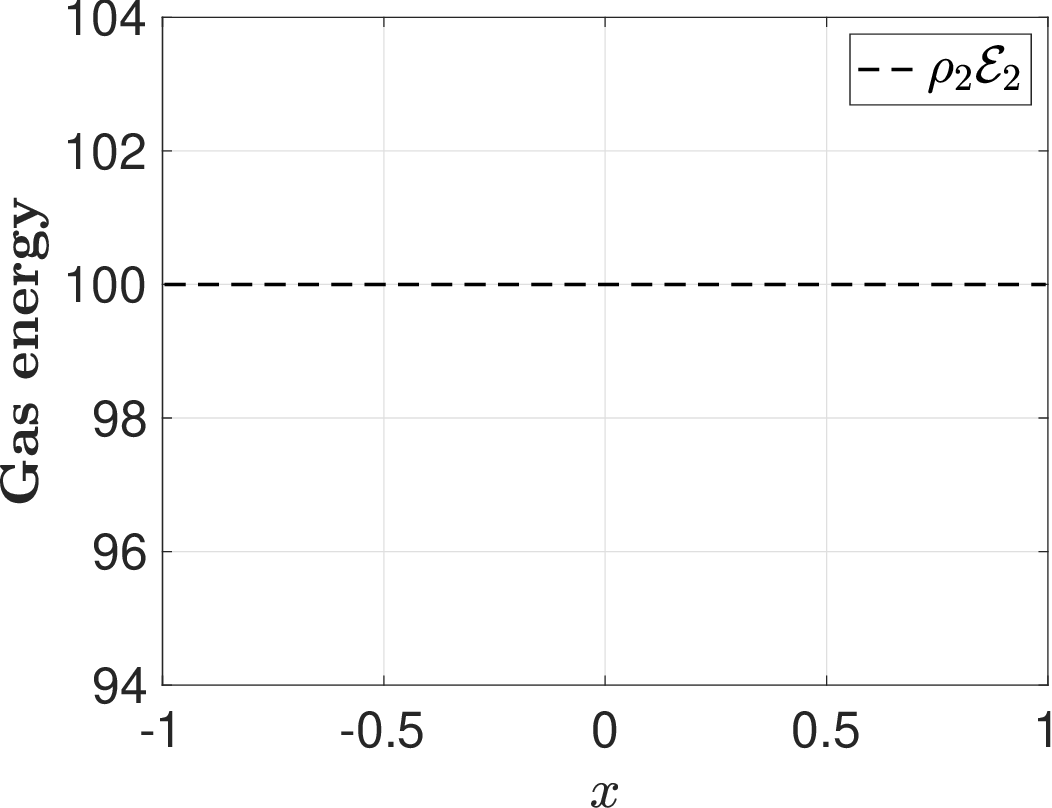}}
{\includegraphics[width = 0.2 \textwidth]{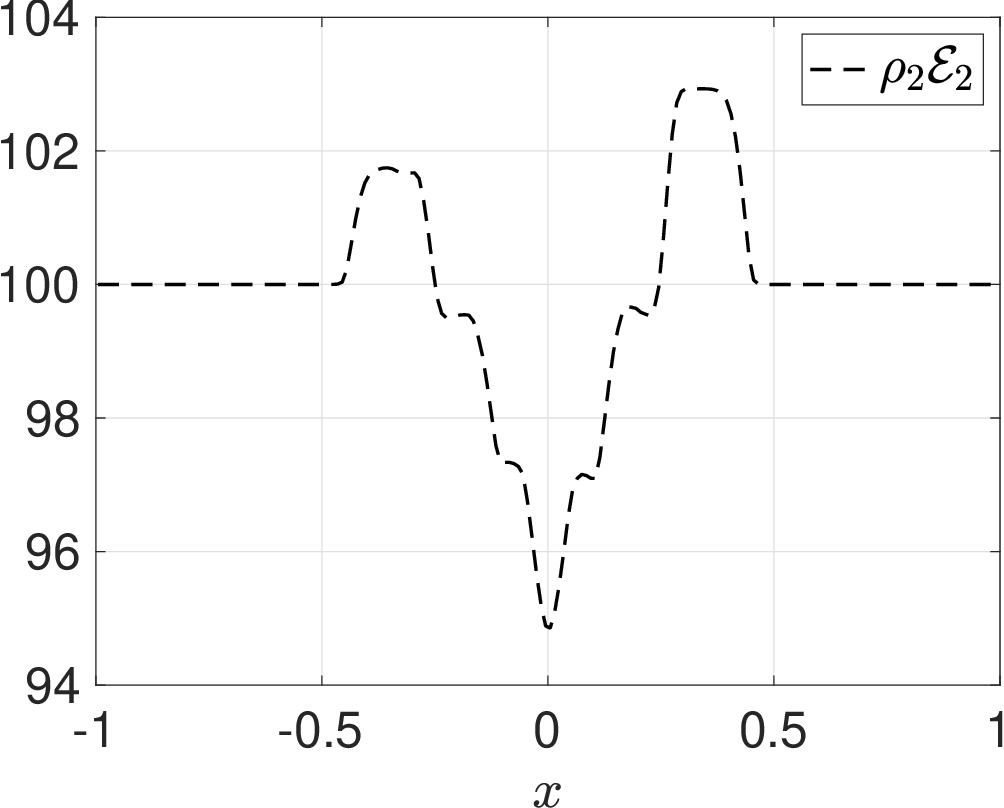}}
{\includegraphics[width = 0.2 \textwidth]{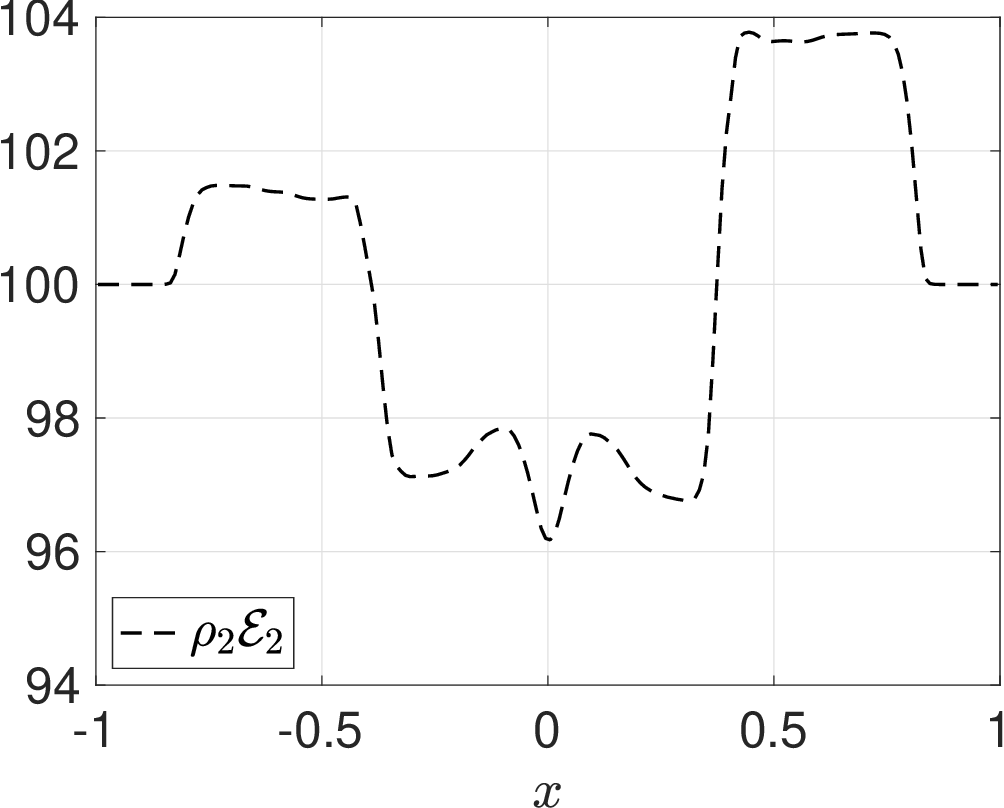}}
{\includegraphics[width = 0.2 \textwidth]{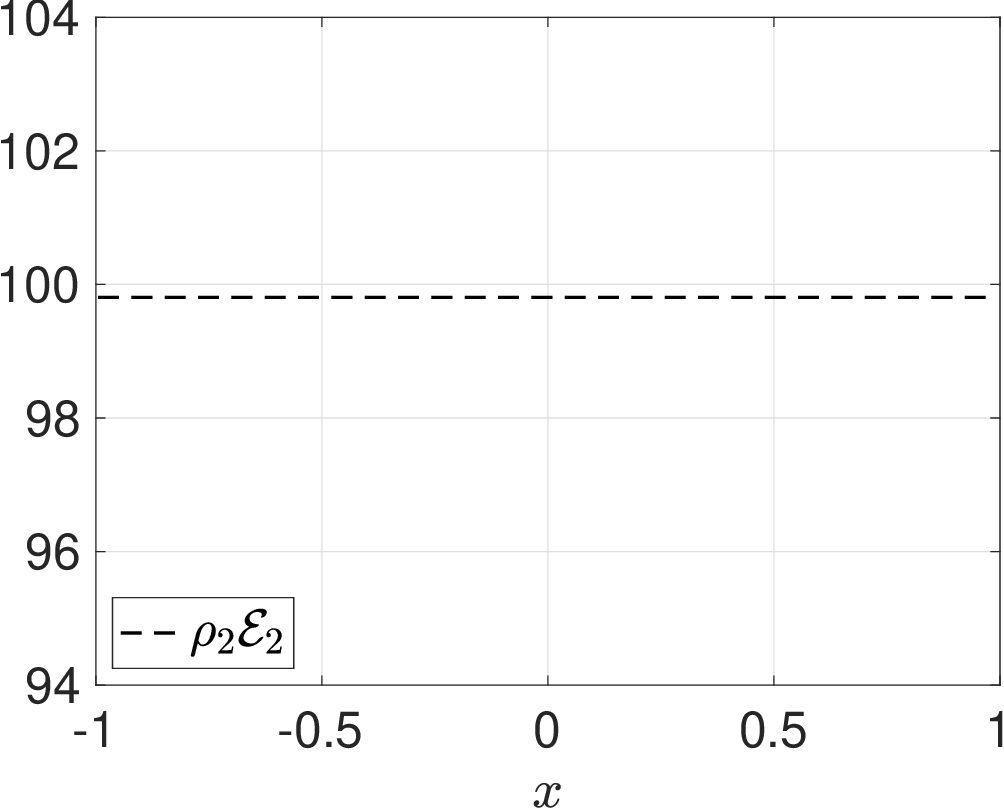}}\\
\caption{Top panels: 3D view of the duct's geometry and interface at times $t = 0.01$ (left) and $t= 5$ (right), using $200$ grid points. The rest of the panels are organized as follows. From top to bottom: topography and interface,  velocities, and gas energy density. From left to right: solutions at times $t=0,0.03,0.08,50$. The geometry, topography, and initial data are defined by equations \eqref{cylinder_shrinks}, \eqref{Smoothbottom}, and \eqref{eq:IC_pgas} respectively.}
\label{fig:WBtoSS}
\end{figure}

Here we consider an air–water configuration, for which the density difference between the two layers is large. We expect that perturbations in the liquid layer have a significant impact on the dynamics of the gas in the upper layer. The purpose of this numerical test is to assess and confirm this behavior. The duct's geometry in this case is given by 
\begin{eqnarray} \label{cylinder_shrinks}
\sigma(x,z) = 2 \sqrt{\max\left(0,\ r_{0}^2 - (r_{0} - z)^2\right)}, \quad \text{with } r_{0} = 0.55 \left( 1-\frac{1+x}{5} \right),
\end{eqnarray}
and it corresponds to a duct with circular cross sections and a cross-sectional area that shrinks along the axial direction. The continuous bottom topography is given by
\begin{eqnarray} \label{Smoothbottom}
\mathcal{B}(x) = \max\left(\frac{1}{10}\left(1 + \sin\left(\frac{\pi (x - 0.4)}{0.25}\right)\right),\ 0\right).
\end{eqnarray}
Using a perturbation in the depth of water layer, the initial conditions are given by
\begin{eqnarray}
\label{eq:IC_pgas}
\rho_1 =1000 ,\,\,\, u_1(x,0) = u_2(x,0) = 0,\,\,\, \rho_2(x,0) = 1, \nonumber\\
\mathcal{E}_2(x,0)  = 100,\,\,\,
\mathcal{B} + h(x,0) = 
\begin{cases}
0.5 + 10^{-1} \hspace{0.3cm}\hspace{.4cm} \mbox{if}\,\, x \in [-0.2,0.2], \\
0.5  \hspace{2.3cm} \text{otherwise.}  
\end{cases}
\end{eqnarray}

The numerical results are shown in Figure \ref{fig:WBtoSS}, using 200 grid points. The top panels show a 3D view of the duct. The top left panel shows the solution at an early time $t = 0.01$, indicating a perturbation in the interface. We expect that this perturbation propagates and leaves the domain for large $t$.  The top right panel shows the solution (3D view) at time $t=50$, confirming that the interface is flat again. The rest of the panels show details about the solution profiles at different times, as the perturbation travels throughout the domain. From that set, the top panels show the profiles of the topography and the interface, followed by the velocities ($u_1, u_2$) and gas energy ($\rho_2 \mathcal E_2$). From left to right, those profiles are shown at times $t= 0,0.03,0.08,50$. At the final time $t=50$, all the variables are flat again, showing convergence. This is consistent with the fact that steady states of rest in our model satisfy $\mathcal B+h = \text{constant}, u_1 = u_2 = 0, \rho_2 \mathcal E_2 = \text{constant}$, according to equation \eqref{eq:SteadyStatesRest}.

%\newpage

\subsection{Perturbation in the gas layer over a discontinuous bottom topography} 

The initial conditions in this numerical test consist of a flat interface with zero velocities, which in classical shallow water systems correspond to steady states of rest. Only the gas density and gas energy density are perturbed. As it occurs in air-water interactions, the density ratio $\rho_2/\rho_1$ is very small here. As a result, the gas fluctuations are expected to generate very weak momentum and energy exchanges with the liquid in the lower layer. We will verify that in this test.

\begin{figure}
\centering
\begin{subfigure}{0.41\textwidth}
    {\includegraphics[width=.99\textwidth]{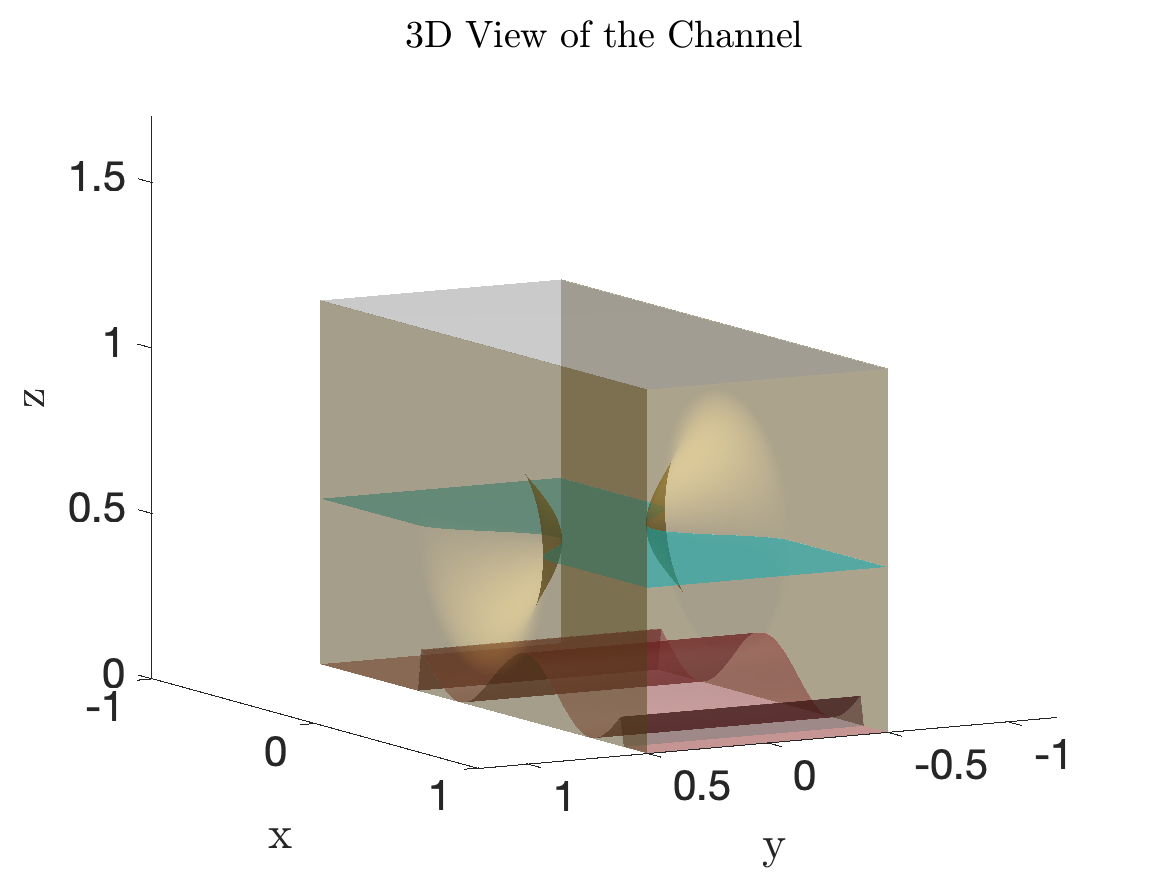}}
    \caption{3D view.}
\end{subfigure}
\begin{subfigure}{0.55\textwidth}
	{\includegraphics[width = 0.32 \textwidth]{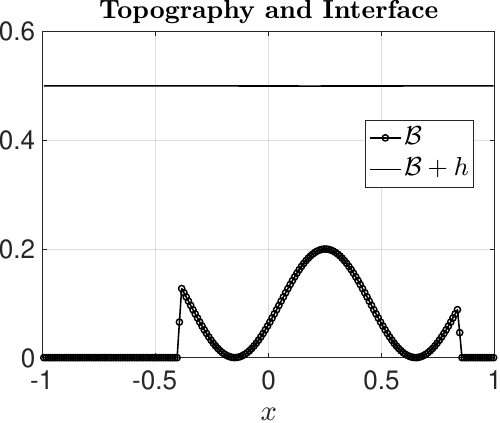}}
	{\includegraphics[width = 0.335\textwidth]{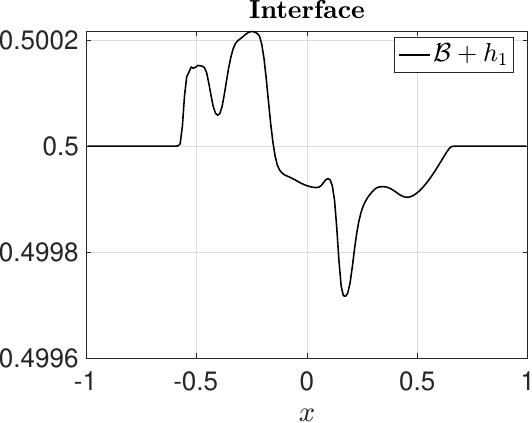}}
	{\includegraphics[width = 0.32\textwidth]{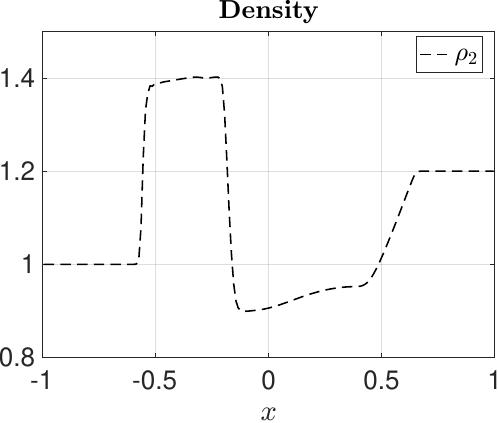}}\\
	{\includegraphics[width = 0.32\textwidth]{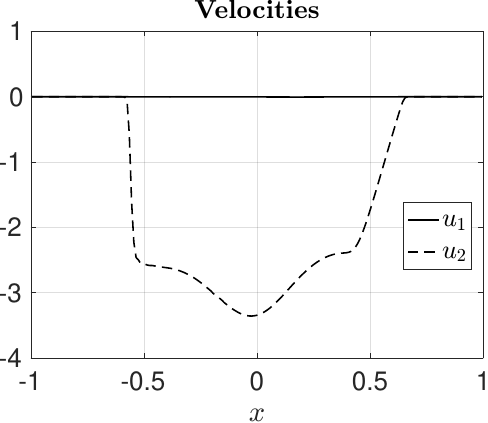}}
	{\includegraphics[width = 0.32\textwidth]{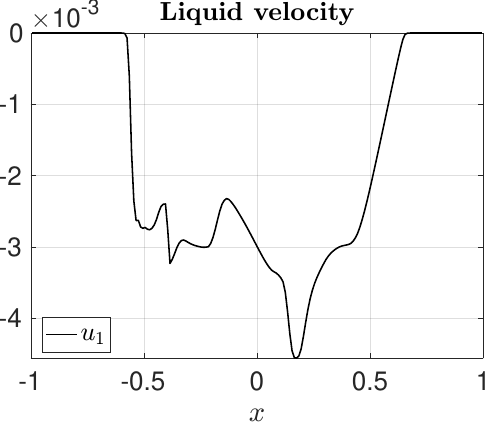}}
	{\includegraphics[width = 0.32 \textwidth]{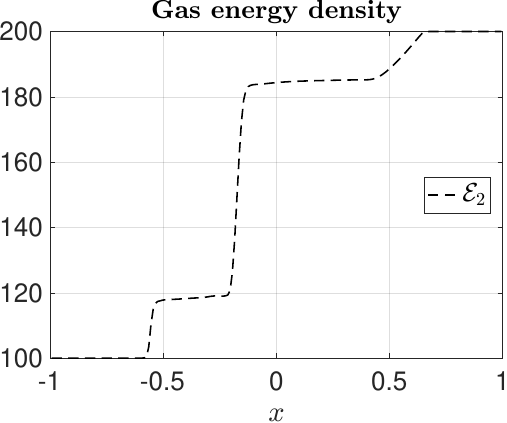}}
    \caption{Profiles.}
\end{subfigure}        
\caption{(a) Three dimensional view of the topography with geometry given by \eqref{charbottom} and \eqref{bumpchannel}. The interface corresponds to the solution at $t= 0.06$ with initial conditions given by \eqref{eq:IC_pgas1}. Section (b): Profiles using $200$ grid points. Top panels: topography and interface (left), interface (middle), and gas density (right). Bottom panels: velocities (left), liquid velocity (middle), and gas energy density (right).}
\label{fig:purt_gas_cyl}
\end{figure}

The discontinuous topography in this numerical test is given by
\begin{eqnarray} \label{charbottom}
\mathcal{B}(x)= 0.1 \left(1 + \cos\left(\frac{\pi(x - 0.25)}{0.4}\right)\right) \cdot \chi_{[-0.4,\ 0.85]}(x).
\end{eqnarray}
The channel's geometry is defined by
\begin{eqnarray} \label{bumpchannel}
\sigma(x,z) =
\begin{cases}
1 - \dfrac{0.6}{2} \left(1 + \cos\left( \dfrac{\pi r(x,z)}{0.4} \right) \right), & \text{if } r(x,z) \leq 0.4, \\
1, & \text{otherwise},
\end{cases}
\label{geometry4}
\end{eqnarray} where
\[
r(x,z) = \sqrt{(x - 0)^2 + (z - 0.5)^2}.
\]
This width corresponds to a channel with vertical walls and a centered contraction. We also consider a perturbation in the gas density and gas energy density as:
\begin{eqnarray}
\label{eq:IC_pgas1}
\rho_1 =1000 ,\,\,\, u_1(x,0) = u_2(x,0) = 0, \,\,\, \mathcal{B}(x)+h(x,0) = 0.5, \nonumber\\
\rho_2(x,0) = \begin{cases}
1 \hspace{.1cm} \text{ if }   -1 \le x < 0, \\
1.2  \hspace{0.1cm} \text{ if }   0 \le x \le 1.
\end{cases},\, \mathcal{E}_2(x,0) = \begin{cases}
100 \hspace{.1cm} \text{ if }   -1 \le x < 0, \\
200  \hspace{.1cm} \text{ if }   0 \le  x \le 1.
\end{cases}
\end{eqnarray}

The numerical approximations at $t=0.06$ are shown in Figure \ref{fig:purt_gas_cyl}. The left panel (a) shows a 3D view of the duct. The discontinuous topography is clearly identified in brown color and no variations are visible in the interface. This is because the initial conditions in equation \eqref{eq:IC_pgas1} consist of variations in the gas layer, which does not have a significant influence on the liquid layer due to the small gas density compared to the liquid counterpart. Under this pronounced density difference, the momentum and energy exchanges between layers due to gas variations are negligible. This can be corroborated in section (b), panels top left and top  middle. The variations in the interface are  less than 0.03 \%. The gas density (top right) and gas energy density (bottom right) show a similar structure to typical solutions of Riemann problems, influenced in addition by the presence of non-uniform topography and duct's geometry. The liquid and gas velocities are shown in the bottom left panel. As one can see, the liquid velocity is much smaller compared to the gas counterpart. In the bottom middle panel one can observe that the liquid velocity variations are less than $4.5 \times 10^{-3}$. This is consistent with the weak influence of the gas perturbations over the liquid layer when the density ratio $\rho_2/\rho_1$ is small.

%\newpage

\subsection{Perturbation in gas layer with liquid and gas hydrogen} 

\begin{figure}
\centering
\begin{subfigure}{0.49\textwidth}
    {\includegraphics[width=.99\textwidth]{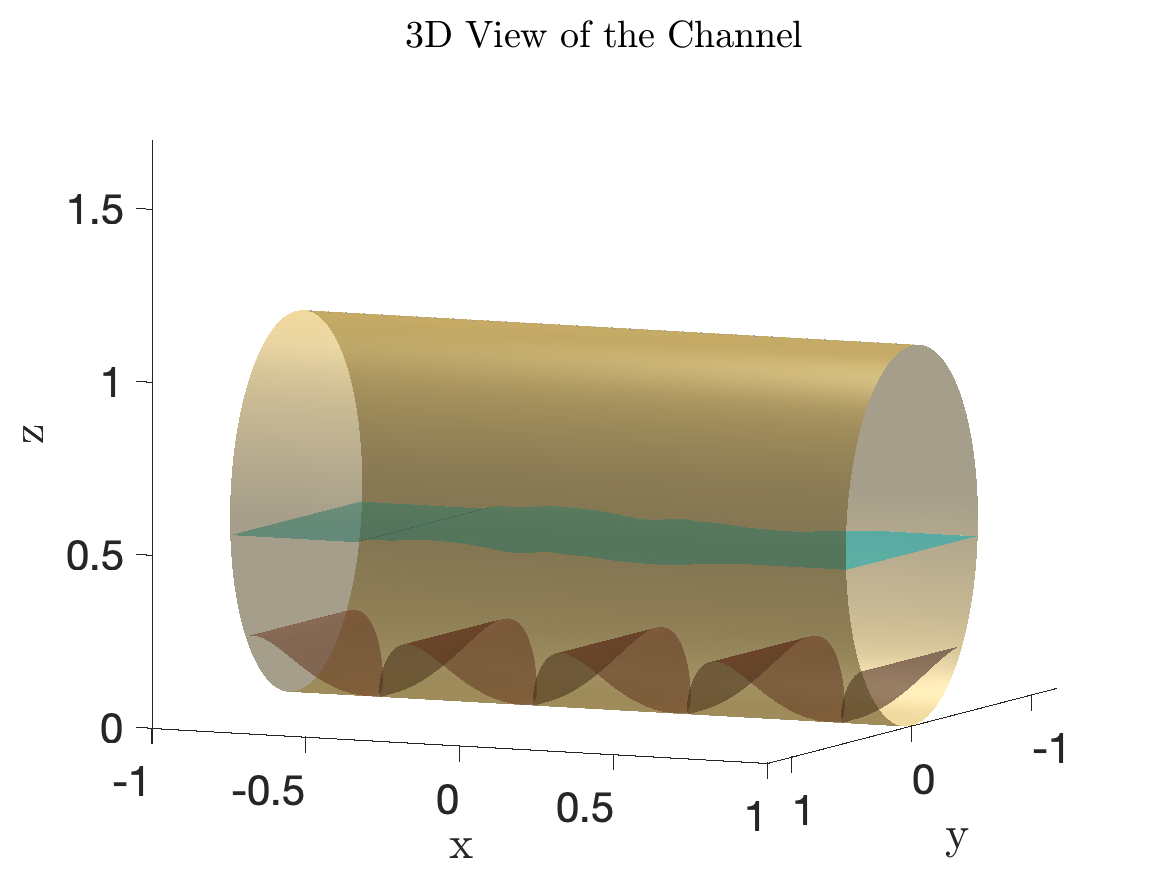}}
    \caption{3D view.}
\end{subfigure}
\begin{subfigure}{0.45\textwidth}
	{\includegraphics[width = 0.47 \textwidth]{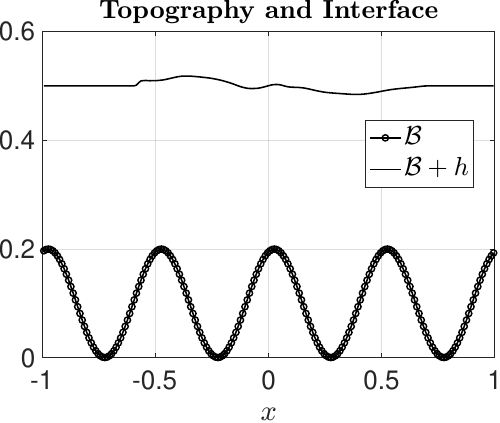}}
	{\includegraphics[width = 0.47\textwidth]{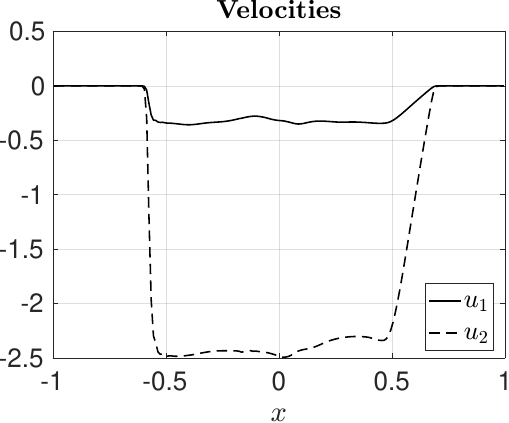}}\\
	{\includegraphics[width = 0.47 \textwidth]{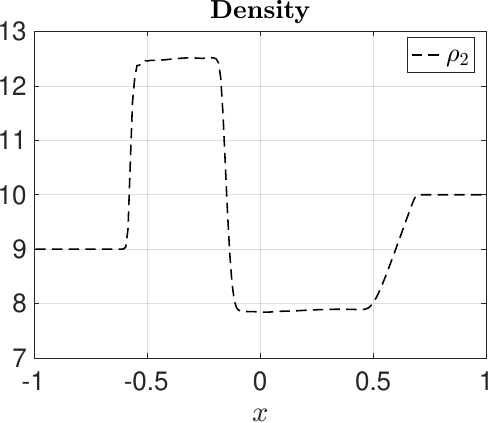}}
	{\includegraphics[width = 0.48 \textwidth]{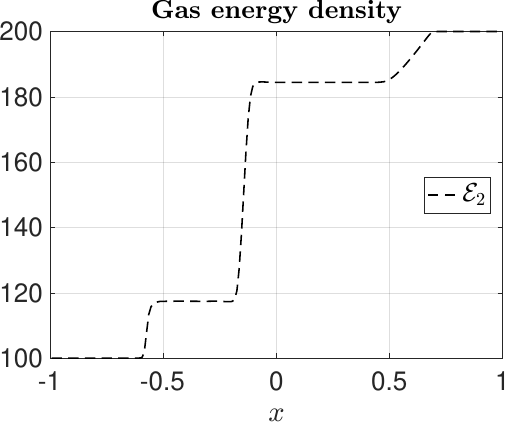}}
    \caption{Profiles.}
\end{subfigure}        
\caption{Panel (a): Three dimensional view of the duct's geometry and topography given by \eqref{Smoothbottom} and \eqref{cylinder}. Sección (b): Profiles at time $t= 0.06$ (using $200$ grid points) described the initial conditions given by equation \eqref{eq:IC_liqhydr}. Top left: Topography and interface. Top right: velocities. Bottom left: Gas density. Bottom right: Gas energy. } 
\label{fig:3DlHgH}
\end{figure}

In this numerical test, the density difference between the gas and liquid layers is not as significan as in previous numerical tests. This may occur, for instance, in interactions between gaseous and liquid hydrogen. Experiments involving those two phases of hydrogen have shown to be relevant, for instance, in the analysis of cryogenic boiling. See \cite{yuan2008cryogenic} for more details on this topic. In terrestrial conditions, the two phases become stratified: liquid stays along the bottom wall (heavier fluid) while vapor occupies the upper portion of the tube (lighter fluid).

Although we do not consider temperature fluctuations as it occurs in cryogenic boiling, we will consider here a test where the corresponding densities of the fluid are like those involving hydrogen. We will perturb the gas upper layer and observe the influence of those perturbations on the liquid bottom layer. In this numerical test, the width is given by equation \eqref{cylinder} for a cylindrical duct. The model's parameter values are  $g = 9.8$, and $\gamma_2 = 1.405$. The initial conditions given in the following equations consider a perturbation in $\mathcal{E}_2$, and $\rho_2$:
\begin{eqnarray}
\label{eq:IC_liqhydr}
\rho_1 =70 ,\,\,\, u_1(x,0) = u_2(x,0) = 0, \, \mathcal B(x)+ h(x,0) = 0.5,\nonumber\\
\rho_2(x,0) = \begin{cases}
9 \hspace{.5 cm} \text{ if }   x < 0 \\
10  \hspace{0.5cm} \text{ if }   x \ge 0
\end{cases},\, \mathcal{E}_2(x,0) =\mathcal{E}_2^{(0)} = \begin{cases}
100 \hspace{.2 cm} \text{ if }   x < 0 \\
200  \hspace{.2 cm} \text{ if }   x \ge 0
\end{cases}.
\end{eqnarray}

Figure \ref{fig:3DlHgH} shows the numerical results. The left panel (a) shows the 3D view of the channel at time $t=0.06$. Even in the 3D view, perturbations are visible in the interface. This can be corroborated in the profiles displayed in section (b) of the figure. The top left panel shows the interface and topography while the top right one displays both gas and liquid velocities. The gas density is shown in the bottom left panel and the gas energy density can be seen in the bottom right panel. The gas velocity is larger in magnitude when compared to the liquid velocity. However, the liquid velocity is non negligible anymore, causing also disturbances in the interface. In summary, this numerical test shows that the influence of gas variations can cause significant liquid variations when the density difference is not so pronounced. \\

\noindent
{\bf Conclusions.}
In this work, we derived and analyzed a model for two-layer immiscible gas–liquid flows in pipes with general cross sections. The resulting system of balance laws is conditionally hyperbolic and consists of two subsystems. The subsystem associated with the lower layer resembles the shallow water approximation for an incompressible fluid with hydrostatic pressure, while the upper layer is governed by an ideal gas law. Both layers are coupled through non-conservative products that represent momentum and energy exchanges between layers. We also describe important properties of the model, including the existence of an entropy pair.

The model accurately captures different flow regimes. For instance, when the density difference between layers is large (e.g., water and air), perturbations in the gas layer do not lead to significant momentum and energy exchanges between layers, and the coupling is mainly one-way: the liquid layer has a strong influence on the gas dynamics, whereas the influence of the gas on the liquid is negligible. In contrast, when the density difference is not too large, momentum and energy exchanges between layers become significant in both directions. The numerical results presented confirm this behavior, along with other key properties of the model.

%\newpage

\section*{Acknowledgement}
S. Shah was partially supported by the Office of Naval Research (ONR) under Award NO: N00014-24-1-2147, and 
the Air Force Office of Scientific Research (AFOSR) under Award NO: FA9550-25-1-0231. G. Hernández-Dueñas was supported in part by grant DGAPA-PAPIIT-UNAM IN115925. %We would like to thank the very useful comments and conversations with Smadar Karni. 

\bibliography{References}

\end{document}